\documentclass[a4paper,english,numberwithinsect]{lipics-v2021}

\usepackage{microtype} 

\graphicspath{{./figures/}}

\title{On the Shape Containment Problem within the Amoebot Model with Reconfigurable Circuits}
\titlerunning{On the Shape Containment Problem within the Amoebot Model}

\author{Matthias Artmann}{Paderborn University, Germany}{matthias.artmann@uni-paderborn.de}{https://orcid.org/0009-0006-4530-2303}{}
\author{Andreas Padalkin}{Paderborn University, Germany}{andreas.padalkin@upb.de}{https://orcid.org/0000-0002-4601-9597}{}
\author{Christian Scheideler}{Paderborn University, Germany}{scheideler@upb.de}{https://orcid.org/0000-0002-5278-528X}{}

\authorrunning{M. Artmann, A. Padalkin, and C. Scheideler}
\ArticleNo{100}

\ccsdesc{Theory of computation~Distributed computing models}
\ccsdesc{Theory of computation~Computational geometry}

\keywords{programmable matter, amoebot model, reconfigurable circuits, shape containment}

\funding{This work was supported by the DFG Project SCHE 1592/10-1.}
\acknowledgements{We want to thank Daniel Warner for his guidance and helpful discussions.}
\pdfoutput=1
\hideLIPIcs
\nolinenumbers

\usepackage{bm}

\newcommand*{\eg}{e.\,g.}
\newcommand*{\ie}{i.\,e.}

\newcommand*{\whp}{w.h.p.}
\newcommand*{\etc}{etc.}
\newcommand*{\Wlog}{w.l.o.g.}

\newcommand{\cb}{,\allowbreak}

\newcommand*{\PASC}{PASC}

\newcommand*{\Nats}{\mathbb{N}}
\newcommand*{\Ints}{\mathbb{Z}}

\newcommand*{\Reals}{\mathbb{R}}

\newcommand{\Set}[1]{\{ #1 \}}

\newcommand{\SetBar}[2]{\{ #1 \mid #2 \}}
\newcommand{\SetBarL}[2]{\left\{ #1 ~\middle|~ #2 \right\}}
\newcommand{\Abs}[1]{| #1 |}

\newcommand{\Ceil}[1]{\lceil #1 \rceil}

\newcommand{\Floor}[1]{\lfloor #1 \rfloor}

\newcommand{\BigO}[1]{\mathcal{O}\,( #1 )}

\newcommand{\BigOmega}[1]{\Omega\,( #1 )}

\newcommand{\BigTheta}[1]{\Theta\,( #1 )}

\newcommand{\LittleO}[1]{\mathrm{o}\,( #1 )}

\newcommand{\Directions}{\mathcal{D}}
\newcommand{\E}{\textsc{E}}
\newcommand{\NNE}{\textsc{NE}}
\newcommand{\NNW}{\textsc{NW}}
\newcommand{\W}{\textsc{W}}
\newcommand{\SSW}{\textsc{SW}}
\newcommand{\SSE}{\textsc{SE}}

\newcommand{\UVec}[1]{\bm{u}_{ #1 }}
\newcommand{\Side}[1]{\bm{\mathrm{ #1 }}}

\newcommand{\Geqt}{G_{\Delta}}
\newcommand{\Veqt}{V_{\Delta}}
\newcommand{\Eeqt}{E_{\Delta}}
\newcommand{\PinCfg}{\mathcal{Q}}

\newcommand{\kmax}{k_{\mathrm{max}}}
\newcommand{\VP}{\mathcal{V}}
\newcommand{\Line}{\mathrm{L}}
\newcommand{\Tri}{\mathrm{T}}
\newcommand{\Path}{\Pi}
\newcommand{\Map}{\Phi}

\newcommand{\NodeType}{\tau}
\newcommand{\NodeDir}{d}
\newcommand{\NodeLen}{\ell}
\newcommand{\TLine}{\Line}
\newcommand{\TTri}{\Tri}
\newcommand{\TUnion}{\cup}
\newcommand{\TSum}{\oplus}
\newcommand{\TShift}{+}

\begin{document}

\maketitle

\begin{abstract}
In \emph{programmable matter}, we consider a large number of tiny, primitive computational entities called \emph{particles} that run distributed algorithms to control global properties of the particle structure.
\emph{Shape formation} problems, where the particles have to reorganize themselves into a desired shape using basic movement abilities, are particularly interesting.
In the related \emph{shape containment} problem, the particles are given the description of a shape \(S\) and have to find maximally scaled representations of \(S\) within the initial configuration, without movements.
While the shape formation problem is being studied extensively, no attention has been given to the shape containment problem, which may have additional uses beside shape formation, such as detection of structural flaws.

In this paper, we consider the shape containment problem within the \emph{geometric amoebot model} for programmable matter, using its \emph{reconfigurable circuit extension} to enable the instantaneous transmission of primitive signals on connected subsets of particles.
We first prove a lower runtime bound of \(\BigOmega{\sqrt{n}}\) synchronous rounds for the general problem, where \(n\) is the number of particles.
Then, we construct the class of \emph{snowflake} shapes and its subclass of \emph{star convex} shapes, and present solutions for both.
Let \(k\) be the maximum scale of the considered shape in a given amoebot structure.
If the shape is star convex, we solve it within \(\BigO{\log^2 k}\) rounds.
If it is a snowflake but not star convex, we solve it within \(\BigO{\sqrt{n} \log n}\) rounds.
\end{abstract}



\section{Introduction}
\label{sec:introduction}

Programmable matter envisions a material that can change its physical properties in a programmable fashion~\cite{toffoli1993programmable} and act based on sensory information from its environment.
It is typically viewed as a system of many identical micro-scale computational entities called \emph{particles}.
Potential application areas include minimally invasive surgery, maintenance, exploration and manufacturing.
While physical realizations of this concept are on the horizon, with significant progress in the field of micro-scale robotics~\cite{yang2022survey,chafik2024conventional}, the fundamental capabilities and limitations of such systems have been studied in theory using various models~\cite{thalamy2019survey}.

In the \emph{amoebot model} of programmable matter, the particles are called \emph{amoebots} and are placed on a connected subset of nodes in a graph.
The \emph{geometric} variant of the model specifically considers the infinite regular triangular grid graph.
This model has been used to study various problems such as leader election, object coating, convex hull formation and shape formation~\cite{derakhshandeh2014brief,derakhshandeh2015leader,daymude2018runtime,diluna2020shape} (also see~\cite{daymude2023canonical} and the references therein).

To circumvent the natural lower bound of \(\BigOmega{D}\) for many problems in the amoebot model, where \(D\) is the diameter of the structure, we consider the \emph{reconfigurable circuit extension} to the model.
In this extension, the amoebots are able to construct simple communication networks called \emph{circuits} on connected subgraphs and broadcast primitive signals on these circuits instantaneously.
This has been shown to accelerate amoebot algorithms significantly, allowing polylogarithmic solutions for problems like leader election, consensus, shape recognition and shortest path forest construction~\cite{feldmann2022coordinating,padalkin2022structural,padalkin2024polylogarithmic}.

The \emph{shape formation problem}, where the initial structure of amoebots has to reconfigure itself into a given target shape, is a standard problem of particular interest~\cite{thalamy2019survey}.
As a consequence, related problems that can lead to improved shape formation solutions are interesting as well.
In this paper, we study the related \emph{shape containment problem}:
Given the description of a shape \(S\), the amoebots have to determine the maximum scale at which \(S\) can be placed within their structure and identify all valid placements at this scale.
A solution to this problem can be extended into a shape formation algorithm by \emph{self-disassembly}, \ie{}, disconnecting all amoebots that are not part of a selected placement of the shape from the structure~\cite{gilpin2012shape,gauci2018programmable}.
The problem can also be interpreted as a discrete variant of the \emph{polygon containment problem} in classical computational geometry, which has been studied extensively in various forms~\cite{chazelle1983polygon,sharir1994extremal}.
However, to our best knowledge, there are no distributed solutions where no single computing unit has the capacity to store both polygons in memory.


\subsection{Geometric Amoebot Model}
\label{subsec:amoebot_model}

We use the \emph{geometric amoebot model} for programmable matter, as proposed in~\cite{derakhshandeh2014brief}.
Using the terminology from the recent \emph{canonical} model description~\cite{daymude2023canonical}, we assume \emph{common direction} and \emph{chirality}, \emph{constant-size memory} and a fully synchronous scheduler, making it \emph{strongly fair}.
We describe the model in sufficient detail here and refer to~\cite{derakhshandeh2014brief,daymude2023canonical} for more information.

The geometric amoebot model places \(n\) particles called \emph{amoebots} on the infinite regular triangular grid graph \(\Geqt = (\Veqt, \Eeqt)\) (see Fig.~\ref{subfig:amoebots:amoebots}).
Each amoebot occupies one node and each node is occupied by at most one amoebot.
We identify each amoebot with the grid node it occupies to simplify the notation.
Thereby, we define the \emph{amoebot structure} \(A \subset \Veqt\) as the subset of occupied nodes.
We assume that \(A\) is finite and its induced subgraph \(G_A := \Geqt|_A = (A, E_A)\) is connected.
Two amoebots are \emph{neighbors} if they occupy adjacent nodes.
Due to the structure of the grid, each amoebot has at most six neighbors.

Each amoebot has a local \emph{compass} identifying one of its incident grid edges as the East direction and a \emph{chirality} defining its local sense of rotation.
We assume that both are shared by all amoebots.
This is not a very restrictive assumption because a common compass and chirality can be established efficiently using circuits~\cite{feldmann2022coordinating}.
Computationally, the amoebots are equivalent to (randomized) finite state machines with a \emph{constant number of states}.
In particular, the amount of memory per amoebot is constant and independent of the number of amoebots in the structure.
This means that, for example, unique identifiers for all amoebots cannot be stored.
All amoebots are identical and start in the same state.
The computation proceeds in \emph{fully synchronous rounds}.
In each round, all amoebots act and change their states simultaneously based on their current state and their received signals (see Section~\ref{subsec:circuit_extension}).
The execution of an algorithm terminates once all amoebots reach a terminal state in which they do not perform any further actions or state transitions.
We measure the time complexity of an algorithm by the number of rounds it requires to terminate.
If the algorithm is randomized, \ie{}, the amoebots can make probabilistic decisions, a standard goal is to find runtime bounds that hold \emph{with high probability} (\whp{})\footnote{An event holds \emph{with high probability} (\whp{}) if it holds with probability at least \(1 - n^{-c}\), where the constant \(c\) can be made arbitrarily large.}.
The algorithms we present in this paper are not randomized.

\begin{figure}[tb]
    \centering
    \begin{minipage}[t]{0.3\textwidth}
        \centering
        \includegraphics[width=\textwidth]{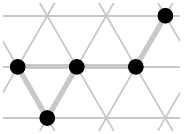}
        \subcaption{Amoebot structure.}
        \label{subfig:amoebots:amoebots}
    \end{minipage}
    \begin{minipage}[t]{0.5\textwidth}
        \centering
        \includegraphics[width=\textwidth]{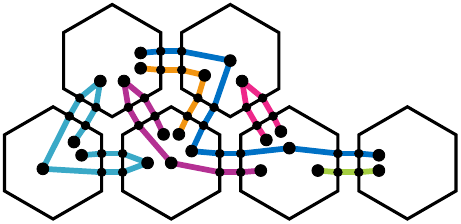}
        \subcaption{Reconfigurable circuit extension.}
        \label{subfig:amoebots:circuits}
    \end{minipage}
    \begin{minipage}[t]{0.5\textwidth}
        \centering
        \includegraphics[width=\textwidth]{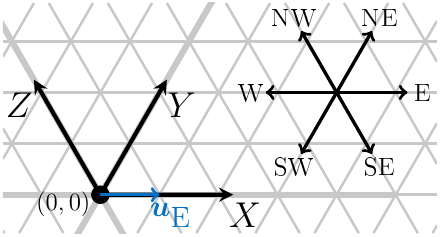}
        \subcaption{Grid axes and cardinal directions.}
        \label{subfig:amoebots:directions}
    \end{minipage}
    \caption{(\subref{subfig:amoebots:amoebots}) shows an amoebot structure in the triangular grid.
        Amoebots are represented by black nodes and neighboring amoebots are connected by thick edges.
        (\subref{subfig:amoebots:circuits}) illustrates the reconfigurable circuit extension.
        Amoebots are drawn as hexagons, pins are black circles on their borders and partition sets are drawn as black circles inside the hexagons.
        The partition sets are connected to the pins they contain.
        Partition sets in the same circuit have lines of the same color.
        (\subref{subfig:amoebots:directions}) shows the axes and cardinal directions in the triangular grid and the unit vector in the East direction.}
    \label{fig:amoebots}
\end{figure}


\subsection{Reconfigurable Circuit Extension}
\label{subsec:circuit_extension}

The \emph{reconfigurable circuit extension}~\cite{feldmann2022coordinating} models the communication between amoebots by placing \(k\) \emph{external links} on each edge connecting two neighbors \(u, v \in A\).
Each external link acts as a communication channel between \(u\) and \(v\), with each amoebot owning one end point of the link.
We call these end points \emph{pins} and assume that the amoebots have a common labeling of their pins and incident links.
The design parameter \(k\) can be chosen arbitrarily but is constant for an algorithm.
\(k = 2\) is sufficient for all algorithms in this paper.

Let \(P(u)\) denote the set of pins belonging to amoebot \(u \in A\).
The state of each amoebot now contains a \emph{pin configuration} \(\PinCfg(u)\), which is a disjoint partitioning of \(P(u)\), \ie{}, the elements \(Q \in \PinCfg(u)\) are pairwise disjoint subsets of pins such that \(\bigcup_{Q \in \PinCfg(u)} Q = P(u)\).
We call the elements \emph{partition sets} and say that two partition sets \(Q \in \PinCfg(u), Q' \in \PinCfg(v)\) of neighbors \(u, v \in A\) are \emph{connected} if there is an external link with one pin in \(Q\) and one pin in \(Q'\).
Let \(\PinCfg := \bigcup_{u \in A} \PinCfg(u)\) be the set of all partition sets in the amoebot structure and let \(E_\PinCfg := \SetBar{\Set{Q, Q'}}{Q \text{~and~} Q' \text{~are connected}}\) be the set of their connections.
Then, we call each connected component \(C\) of the graph \(G_\PinCfg := (\PinCfg, E_\PinCfg)\) a \emph{circuit} (see Fig.~\ref{subfig:amoebots:circuits}).
An amoebot \(u\) is part of a circuit \(C\) if \(C\) contains at least one partition set of \(u\).
Note that multiple partition sets of an amoebot \(u\) may be contained in the same circuit without \(u\) being aware of this due to its lack of global information.
Also observe that if every partition set in \(\PinCfg\) is a singleton, \ie{}, only contains a single pin, then each circuit in \(G_\PinCfg\) only connects two neighboring amoebots, allowing them to exchange information locally.

During its activation, each amoebot can modify its pin configuration arbitrarily and send a primitive signal called a \emph{beep} on any selection of its partition sets.
A beep is broadcast to the circuit containing the partition set it was sent on.
It is available to all partition sets in that circuit in the next round.
An amoebot can tell which of its partition sets have received a beep but it has no information on the identity, location or number of beep origins.


\subsection{Problem Statement}
\label{subsec:problem_statement}

Consider the embedding of the triangular grid graph into \(\Reals^2\) such that the grid's faces form equilateral triangles of unit side length, one grid node is placed at the plane's \emph{origin} \((0, 0) \in \Reals^2\) and one grid axis aligns with the \(x\)-axis.
We define this axis as the grid's \(X\) axis and call its positive direction the East (\(\E\)) direction.
Turning in counter-clockwise direction, we define the other grid axes as the \(Y\) and \(Z\) axes and identify their positive directions as the North-East (\(\NNE\)) and North-West (\(\NNW\)) directions, respectively.
Let the resulting set of directions be the \emph{cardinal directions} \(\Directions = \Set{\E\cb \NNE\cb \NNW\cb \W\cb \SSW\cb \SSE}\).
We denote the unit vector in direction \(d \in \Directions\) by \(\UVec{d}\).
See Fig.~\ref{subfig:amoebots:directions} for illustration.

A \emph{shape} \(S \subset \Reals^2\) is defined as the finite union of some of the embedded grid's nodes, edges and triangular faces (see Fig.~\ref{fig:shape_def}).
An edge contains its two end points and a face contains its three enclosing edges.
Shapes must be connected subsets of \(\Reals^2\) but we allow them to have \emph{holes}, \ie{}, \(\Reals^2 \setminus S\) might not be connected.
This shape definition matches the one used in~\cite{diluna2020shape} for shape formation and extends the definition used in~\cite{feldmann2022coordinating} for shape recognition.

Two shapes are \emph{equivalent} if one can be obtained from the other by a rigid motion, \ie{}, a composition of a translation and a rotation.
Only rotations by multiples of \(60^\circ\) and translations by integer distances along the grid axes yield valid shapes because the shape's faces and edges must align with the grid.
We denote rotated versions of a shape \(S\) by \(S^{(r)}\), where \(r \in \Ints\) is the number of counter-clockwise \(60^\circ\) rotations around the origin.
Note that \(r \in \Set{0, \ldots, 5}\) is sufficient to represent all distinct rotations.
For \(t \in \Reals^2\), we denote \(S\) translated by \(t\) by \(S + t := \SetBar{p + t}{p \in S}\).
This is a valid shape if and only if \(t\) is the position of a grid node.
Let \(S\) be a shape and \(k \in \Reals\) be a \emph{scale factor}, then we define \(k \cdot S := \SetBar{k \cdot p}{p \in S}\) to be the shape \(S\) scaled by \(k\).
We only consider positive integer scale factors to ensure that the resulting set is a valid shape.
If \(S\) is \emph{minimal}, \ie{}, there is no scale factor \(0 < k' < 1\) such that \(k' \cdot S\) is a valid shape, then the integer scale factors cover all possible scales of \(S\) that produce valid shapes (see Lemma~1 in~\cite{diluna2020shape}).

\begin{figure}
	\centering
	\begin{minipage}{0.4\textwidth}
		\centering
		\includegraphics[width=\textwidth]{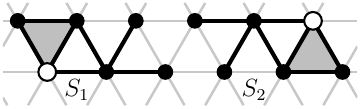}
	\end{minipage}
	\begin{minipage}{0.55\textwidth}
		\centering
		\includegraphics[width=\textwidth]{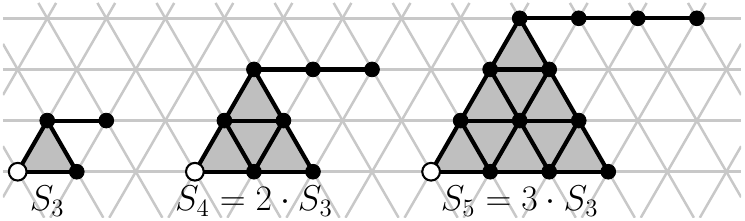}
	\end{minipage}
	\caption{Examples of equivalent and scaled shapes.
		Each shape is identified by the grid nodes, edges and faces it contains.
		The origin of each shape is highlighted in white (we place the shapes at different locations for convenience).
		\(S_1\) and \(S_2\) are equivalent and each contains one face and one hole.
		The shapes \(S_3\), \(S_4\) and \(S_5\) illustrate the scaling operation.}
	\label{fig:shape_def}
\end{figure}

Let \(V(S) \subset \Veqt\) denote the set of grid nodes covered by \(S\).
For convenience, we assume that all shapes contain the origin node, which ensures that a shape does not move relative to the origin when it is rotated or scaled and that the union \(S_1 \cup S_2\) of shapes is always connected.
Let \(A\) be an amoebot structure and \(S\) a shape containing the origin.
We say that an amoebot \(p \in A\) represents a \emph{valid placement} of \(S\) in \(A\) if \(V(S + p) \subseteq A\), where we abuse the notation further to let \(p\) denote the vector in \(\Reals^2\) pointing to amoebot (or node) \(p\).
Let \(\VP(S, A) \subseteq A\) denote the set of valid placements of \(S\) in \(A\).
The \emph{maximum scale} of \(S\) in \(A\) is the largest scale \(k \in \Nats_0\) such that there is a valid placement of \(k \cdot S^{(r)}\) in \(A\) for some \(r \in \Ints\):
\[
    \kmax = \kmax(S, A) := \sup\SetBarL{k \in \Nats_0}{\exists~r \in \Ints: \VP(k \cdot S^{(r)}, A) \neq \emptyset}
\]
\(\kmax\) is well-defined because \(0 \cdot S\) is a single node for every shape \(S\), which fits into any non-empty amoebot structure \(A\).
We obtain \(\kmax = \infty\) if and only if every \(k \in \Nats_0\) has a valid placement.
This only happens for trivial shapes, \ie{}, the empty shape and the shape that is only a single node, which we will not consider further.

We define the \emph{shape containment problem} as follows:
Let \(S\) be a shape (containing the origin).
An algorithm solves the shape containment problem instance \((S, A)\) for amoebot structure \(A\) if it terminates eventually and at the end, either
\begin{enumerate}
	\item all amoebots know that the maximum scale is \(0\) if this is the case, or
	\item for every \(r \in \Set{0, \ldots, 5}\), each amoebot knows whether it is contained in \(\VP(\kmax \cdot S^{(r)}, A)\).
\end{enumerate}
The algorithm solves the shape containment problem for \(S\) if it solves the shape containment instances \((S', A)\) for all finite connected amoebot structures \(A\), where \(S'\) is equivalent to \(S\), contains the origin and is the same for all instances.

There are two key challenges in solving the shape containment problem.
First, the amoebots have to find the maximum scale \(\kmax\).
We approach this problem by testing individual scale factors for valid placements until \(\kmax\) is fixed.
We call this part of an algorithm the \emph{scale factor search}.
Second, for a given scale \(k\) and a rotation \(r\), the valid placements of \(k \cdot S^{(r)}\) have to be identified.
In our approach, we initially view all amoebots as \emph{placement candidates} and then \emph{eliminate} candidates that can be ruled out as valid placements.
To safely eliminate a candidate \(p\), a proof of an unoccupied node that prevents the placement at \(p\) has to be delivered to \(p\).
This information always originates at the boundaries of the structure, \ie{}, amoebots with less than six neighbors.
A \emph{valid placement search} procedure transfers this information from the boundaries to the rest of the structure.
It has to ensure that an amoebot is eliminated if and only if it does not represent a valid placement.

In this paper, we develop a class of shapes for which the shape containment problem can be solved in sublinear time using circuits.
First, as a motivation, we prove a lower bound for a simple example shape that holds even if the maximum scale is already known, demonstrating a bottleneck for the transfer of elimination proofs.
We then introduce scale factor search methods, solutions for basic line and triangle shapes, and primitives for the efficient transfer of more structured information.
Our main result is a sublinear time algorithm solving the shape containment problem for the class of \emph{snowflake} shapes, which we develop based on these primitives.
We also show that for the subclass of \emph{star convex} shapes, there is even a polylogarithmic solution.


\subsection{Related Work}
\label{subsec:related_work}





The authors of~\cite{feldmann2022coordinating} demonstrated the potential of their reconfigurable circuit extension with algorithms solving the leader election, compass alignment and chirality agreement problems within \(\BigO{\log n}\) rounds, \whp{}
They also presented efficient solutions for some exact shape recognition problems:
Given common chirality, an amoebot structure can determine whether it matches a scaled version of a given shape composed of edge-connected faces in \(\BigO{1}\) rounds.
Without common chirality, convex shapes can be detected in \(\BigO{1}\) rounds and parallelograms with linear or polynomial side ratios can be detected in \(\BigTheta{\log n}\) rounds, \whp{}

The \PASC{} algorithm was introduced in~\cite{feldmann2022coordinating} and refined in~\cite{padalkin2022structural}, and it allows amoebots to compute distances along chains.
It has become a central primitive in the reconfigurable circuit extension, as it was used to construct spanning trees, detect symmetry and identify centers and axes of symmetry in polylogarithmic time, \whp{}~\cite{padalkin2022structural}.
The authors in~\cite{padalkin2024polylogarithmic} used it to solve the single- and multi-source shortest path problems, requiring \(\BigO{\log \ell}\) rounds for a single source and \(\ell\) destinations and \(\BigO{\log n \log^2 k}\) rounds for \(k\) sources and any number of destinations.
The \PASC{} algorithm also plays a crucial role in this paper (see Sec.~\ref{subsubsec:pasc}).

The authors in~\cite{emek2024power} studied the capabilities of a generalized circuit communication model that directly extends the reconfigurable circuit model to general graphs.
They provided polylogarithmic time algorithms for various common graph construction (minimum spanning tree, spanner) and verification problems (minimum spanning tree, cut, Hamiltonian cycle \etc{}).
Additionally, they presented a generic framework for translating a type of lower bound proofs from the widely used CONGEST model into the circuit model, demonstrating that some problems are hard in both models while others can be solved much faster with circuits.
For example, checking whether a graph contains a \(5\)-cycle takes \(\BigOmega{n / \log n}\) rounds in general graphs, even with circuits, while the verification of a connected spanning subgraph can be done with circuits in \(\BigO{\log n}\) rounds \whp{}, which is below the lower bound shown in~\cite{sarma2012distributed}.


In the context of computational geometry, the basic polygon containment problem was studied in~\cite{chazelle1983polygon}, focusing on the case where only translation and rotation are allowed.
The problem of finding the largest copy of a convex polygon inside some other polygon was discussed in~\cite{sharir1994extremal} and~\cite{agarwal1998largest}, for example.
An example for the problem of placing multiple polygons inside another without any polygons intersecting each other is given by~\cite{martins2010simulated}.
More recently, the authors in~\cite{kunnemann2022polygon} showed lower bounds for several polygon placement cases under the \(k\)SUM conjecture.
For example, assuming the \(5\)SUM conjecture, there is no \(\BigO{(p+q)^{3 - \varepsilon}}\)-time algorithm for any \(\varepsilon > 0\) that finds a largest copy of a simple polygon \(P\) with \(p\) vertices that fits into a simple polygon \(Q\) with \(q\) vertices under translation and rotation.
Perhaps more closely related to our setting (albeit centralized) is an algorithm that solves the problem of finding the largest area parallelogram inside of an object in the triangular grid, where the object is a set of edge-connected faces~\cite{alaman2022largest}.


\section{Preliminaries}
\label{sec:preliminaries}

This section introduces elementary algorithms for the circuit extension from previous work.


\subsection{Coordination and Synchronization}
\label{subsec:coordination_synchronization}

As mentioned before, we assume that all amoebots share a common compass direction and chirality.
This is a reasonable assumption because the authors of~\cite{feldmann2022coordinating} have presented randomized algorithms establishing both in \(\BigO{\log n}\) rounds, \whp{}

We often want to synchronize amoebots, for example, when different parts of the structure run independent instances of an algorithm simultaneously.
For this, we can make use of a \emph{global circuit}:
Each amoebot connects all of its pins into a single partition set.
The resulting circuit spans the whole structure and allows the amoebots which are not yet finished with their procedure to inform all other amoebots by sending a beep.
When no beep is sent, all amoebots know that all instances of the procedure are finished.
Due to the fully synchronous scheduler, we can establish the global circuit periodically at predetermined intervals.


\subsection{Chains and Chain Primitives}
\label{subsec:chains}

A \emph{chain} of amoebots with length \(m-1\) is a sequence of \(m\) amoebots \(C = (p_0, \ldots, p_{m-1})\) where all subsequent pairs \(p_i, p_{i+1}\), \(0 \leq i < m-1\), are neighbors, each amoebot except \(p_0\) knows its predecessor and each amoebot except \(p_{m-1}\) knows its successor.
We only consider \emph{simple} chains without multiple occurrences of the same amoebot in this paper.
This makes it especially convenient to construct circuits along a chain, \eg{}, by letting each amoebot on the chain decide whether it connects its predecessor to its successor.

\subsubsection{Binary Operations}
\label{subsubsec:binary_operations}

The constant memory limitation of amoebots makes it difficult to deal with non-constant information, such as numbers that can grow with \(n\).
However, we can use amoebot chains to implement a distributed memory by letting each amoebot on the chain store one bit of a binary number, as demonstrated in~\cite{daymude2020convex,padalkin2022structural}.
Using circuits, we can implement efficient comparisons and arithmetic operations between two operands stored on the same chain.

\begin{lemma}
	\label{lem:binary_operations}
	Let \(C = (p_0, \ldots, p_{m-1})\) be an amoebot chain such that each amoebot \(p_i\) stores two bits \(a_i\) and \(b_i\) of the integers \(a\) and \(b\), where \(a = \sum\limits_{i = 0}^{m-1} a_i 2^i\) and \(b = \sum\limits_{i = 0}^{m-1} b_i 2^i\).
	Within \(\BigO{1}\) rounds, the amoebots on \(C\) can compare \(a\) to \(b\) and compute the first \(m\) bits of \(a + b\), \(a - b\) (if \(a \geq b\)), \(2 \cdot a\) and \(\Floor{a / 2}\) and store them on the chain.
	Within \(\BigO{m}\) rounds, the amoebots on \(C\) can compute the first \(m\) bits of \(a \cdot b\), \(\Floor{a / b}\) and \(a \bmod b\) and store them on the chain.
\end{lemma}
\begin{proof}
    Consider a chain \(C = (p_0, \ldots, p_{m-1})\) storing the two integers \(a = \sum\limits_{i=0}^{m-1} a_i 2^i\) and \(b = \sum\limits_{i=0}^{m-1} b_i 2^i\) such that \(p_i\) holds \(a_i\) and \(b_i\).
	Using singleton circuits, we can compute \(2 \cdot a\) and \(\Floor{a / 2}\) by shifting all bits of \(a\) by one position forwards (towards the successor) or backwards (towards the predecessor) along the chain, which only takes a single round.
	
	Next, as a preparation, we find the \emph{most significant bit} of each number, \ie{}, the largest \(i\) such that \(a_i\) (resp.~\(b_i\)) is \(1\).
	To do this, each amoebot \(p_i\) with \(a_i = 0\) connects its predecessor and successor with a partition set and each amoebot with \(a_i = 1\) sends a beep towards its predecessor.
	This establishes circuits which connect the amoebots storing \(1\)s.
	If amoebot \(p_i\) with \(a_i = 1\) does not receive a beep from its successor, it marks itself as the most significant bit since there is no amoebot \(p_j\) with \(j > i\) and \(a_j = 1\).
	If there is no amoebot storing a \(1\), then \(p_0\) will not receive a beep and can mark itself as the most significant bit.
	We repeat the same procedure for \(b\).
	Both finish after just two rounds.
	Let \(i^\ast\) and \(j^\ast\) be the positions of the most significant bits of \(a\) and \(b\), respectively.
	
	\subparagraph{Comparison}
	To compare \(a\) and \(b\), observe that the largest \(i\) with \(a_i \neq b_i\) uniquely determines whether \(a > b\) or \(a < b\), if it exists.
	The amoebots establish circuits where all \(p_i\) with \(a_i = b_i\) connect their predecessor to their successor and the \(p_i\) with \(a_i \neq b_i\) send a beep towards their predecessor.
	If \(a = b\), no amoebot will send or receive a beep, which is easily recognized.
	Otherwise, let \(k\) be the largest index with \(a_k \neq b_k\).
	Then, \(p_k\) will not receive a beep from its successor but all preceding amoebots will.
	\(p_k\) now locally compares \(a_k\) to \(b_k\) and transmits the result on a circuit spanning the whole chain, \eg{}, by beeping in the next round for \(a > b\) and beeping in the round after that for \(a < b\).
	This only takes two rounds.
	
	\subparagraph{Addition}
	To compute \(c = a + b\), consider the standard written algorithm for integer addition.
	In this algorithm, we traverse the two operands from \(i = 0\) to \(i = m-1\).
	In each step, we compute bit \(c_i\) as the sum of \(a_i\), \(b_i\) and a \emph{carry} bit \(d_i\) originating from the previous operation.
	More precisely, we set \(c_i = (a_i + b_i + d_i) \bmod 2\) and compute \(d_{i+1} = \Floor{(a_i + b_i + d_i) / 2}\).
	Initially, the carry is \(d_0 = 0\).
	Each amoebot \(p_i\) can compute \(c_i\) and \(d_{i+1}\) locally when given \(d_i\).
	Observe the following rules for \(d_{i+1}\):
	If \(a_i = b_i = 0\), we always get \(d_{i+1} = 0\).
	For \(a_i = b_i = 1\), we always get \(d_{i+1} = 1\).
	And finally, for \(a_i \neq b_i\), we get \(d_{i+1} = d_i\).
	These rules allow us to compute all carry bits in a single round:
	All amoebots \(p_i\) with \(a_i \neq b_i\) connect their predecessor to their successor, allowing the carry bit to be forwarded directly through the circuit.
	All other amoebots do not connect their neighbors.
	Now, the amoebots with \(a_i = b_i = 1\) send a beep to their successor.
	All amoebots \(p_i\) with \(d_i = 1\) receive a beep from their predecessor while the other amoebots do not receive such a beep.
	After receiving the carry bits this way, each amoebot computes \(c_i\) locally.
	This procedure requires only two rounds.
	Observe that if \(a + b\) requires more than \(m\) bits, we have \(d_m = 1\), which can be recognized by \(p_{m-1}\).

	\subparagraph{Subtraction}
	To subtract \(b\) from \(a\), we apply the same algorithm as for addition, but with slightly different rules.
	Using the notation from above, the bits of \(c = a - b\) are again computed as \(c_i = (a_i + b_i + d_i) \bmod 2\).
	The rules for computing the carry differ as follows:
	For \(a_i > b_i\), we always get \(d_{i+1} = 0\).
	For \(a_i < b_i\), we always get \(d_{i+1} = 1\).
	Finally, for \(a_i = b_i\), we get \(d_{i+1} = d_i\).
	This is because the carry bit must be subtracted from the local difference rather than added.
	Since the carry bits can be determined just as before, the amoebots can compute \(a - b\) in only two rounds.
	In the case that \(a < b\), \(p_{m-1}\) will recognize \(d_m = 1\) again.
	
	\subparagraph{Multiplication}
	The product \(c = a \cdot b\) can be written as
	\[
	a \cdot b = \sum\limits_{i=0}^{m-1} a_i \cdot 2^i \cdot b = \sum\limits_{i:~ a_i = 1} 2^i \cdot b.
	\]
	We implement this operation by repeated addition.
	Initially, we set \(c = 0\) by letting \(c_i = 0\) for each amoebot \(p_i\).
	In the first step, amoebot \(p_0\) sends a beep on a circuit spanning the whole chain if and only if \(a_0 = 1\).
	In this case, we perform the binary addition of \(c + a_0 \cdot b \cdot 2^0 = c + b\) and store the result in \(c\).
	Otherwise, we keep \(c\) as it is.
	In each following iteration, we move a marker that starts at \(p_0\) one step forward in the chain.
	Before each addition, the amoebot \(p_i\) that holds the marker beeps on the chain circuit if and only if \(a_i = 1\).
	If no beep is sent, the addition is skipped.
	Otherwise, we update \(c \gets c + b'\), where \(b'\) is initialized to \(b\) and its bits are moved one step forward in each iteration.
	The sequence of values of \(b'\) obtained by this is \(b, 2b, 2^2 b, \ldots, 2^{m-1} b\), but limited to the first \(m\) bits.
	Since the higher bits of \(b'\) do not affect the first \(m\) bits of the result, we obtain the first \(m\) bits of \(a \cdot b\).
	Because each iteration only requires a constant number of rounds, the procedure finishes in \(\BigO{m}\) rounds.
	Note that we can already stop after reaching \(a_{i^\ast}\) because all following bits of \(a\) are \(0\), which may improve the runtime if \(i^\ast\) is significantly smaller than \(m\) (\eg{}, constant).
	
	\subparagraph{Division}
	We implement the standard written algorithm for integer division with remainder by repeated subtraction.
	For this, we maintain the division result \(c\), the current divisor \(b'\) and the current remainder \(a'\) as binary counters.
	\(c\) is initialized to \(0\) and \(a'\) and \(b'\) are initialized to \(a\) and \(b\), respectively.
	We start by shifting \(b'\) forward until its most significant bit aligns with that of \(a'\).
	For \(a \geq b\), this succeeds within \(\BigO{m}\) rounds; In case \(a < b\), we can terminate immediately.
	Let \(j\) be the number of steps that were necessary for the alignment.
	After this, each iteration \(i = j, \ldots, 0\) works as follows:
	First, we compare \(a'\) to \(b'\).
	If \(a' < b'\), we keep the bit \(c_i = 0\).
	Otherwise, we record \(c_i = 1\) and compute \(a' \gets a' - b'\).
	At the end of the iteration, we shift \(b'\) back by one step.
	After iteration \(i = 0\), \(c\) contains \(\Floor{a / b}\) and \(a'\) contains the remainder \(a \bmod b\).
	The correctness follows because at the end, \(a' < 2^0 b = b\) and \(a = a' + \sum\limits_{i=0}^j c_i \cdot 2^i \cdot b = a' + c \cdot b\) hold, since in iteration \(i\), \(b'\) is equal to \(2^i b\).
	Because each iteration takes a constant number of rounds and the number of iterations is \(\BigO{m}\), the runtime follows.
\end{proof}

Lemma~\ref{lem:binary_operations} is in fact a minor improvement over the algorithms presented in~\cite{padalkin2022structural}.
Additionally, individual amoebots can execute simple binary operations online on \emph{streams} of bits:

\begin{lemma}
	\label{lem:binary_stream_operations}
	Let \(p\) be an amoebot that receives two numbers \(a, b\) as \emph{bit streams}, \ie{}, it receives the bits \(a_i\) and \(b_i\) in the \(i\)-th iteration of some procedure, for \(i = 0, \ldots, m\).
	Then, \(p\) can compute bit \(c_i\) of \(c = a + b\) or \(c = a - b\) (if \(a \geq b\)) in the \(i\)-th iteration and determine the comparison result between \(a\) and \(b\) by iteration \(m\), with only constant overhead per iteration.
\end{lemma}
\begin{proof}
	Let \(p\) be an amoebot that receives the bits \(a_i\) and \(b_i\) in the \(i\)-th iteration of some procedure, for \(i = 0, \ldots, m\).
	To compute the bits of \(a + b\) and \(a - b\), \(p\) runs the standard written algorithm described above, but sequentially.
	Starting with \(d_0 = 0\), \(p\) only needs access to \(d_i\), \(a_i\) and \(b_i\) to compute \(c_i\) and \(d_{i+1}\) in a single round.
	Because the values from previous iterations do not need to be stored, constant memory is sufficient for this.
	To compare \(a\) and \(b\), \(p\) initializes an intermediate result to "\(=\)" and updates it to "\(<\)" or "\(>\)" whenever \(a_i < b_i\) or \(a_i > b_i\) occurs, respectively.
	Since the relation between \(a\) and \(b\) depends only on the highest value bits that are different, the result will be correct after iteration \(m\).
\end{proof}

\subsubsection{The \PASC{} Algorithm}
\label{subsubsec:pasc}

A particularly useful algorithm in the reconfigurable circuit extension is the \emph{Primary-And-Secondary-Circuit} (\PASC{}) algorithm, first introduced in~\cite{feldmann2022coordinating}.
We omit the details of the algorithm and only outline its relevant properties.
Please refer to~\cite{padalkin2022structural} for details.

\begin{lemma}[\cite{feldmann2022coordinating,padalkin2022structural}]
	\label{lem:pasc}
	Let \(C = (p_0, \ldots, p_{m-1})\) be a chain of \(m\) amoebots.
	The \PASC{} algorithm, executed on \(C\) with start point \(p_0\), performs \(\Ceil{\log m}\) iterations within \(\BigO{\log m}\) rounds.
	In iteration \(j = 0, \ldots, \Ceil{\log m} - 1\), each amoebot \(p_i\) computes the \(j\)-th bit of its distance \(i\) to the start of the chain, \ie{}, \(p_i\) computes \(i\) as a bit stream.
\end{lemma}

The \PASC{} algorithm is especially useful with binary counters.
It allows us to compute the length of a chain, which is received by the last amoebot in the chain and can be stored in binary on the chain itself.
Furthermore, given some binary counter storing a distance \(d\) and some amoebot chain \(C = (p_0, \ldots, p_{m-1})\), each amoebot \(p_i\) can compare \(i\) to \(d\) by receiving the bits of \(d\) on a global circuit in sync with the iterations of the \PASC{} algorithm on \(C\).

\begin{lemma}
	\label{lem:pasc_cutoff}
	Let \(C = (p_0, \ldots, p_{m-1})\) be a chain in an amoebot structure \(A\) and let a value \(d \in \Nats_0\) be stored in some binary counter of \(A\).
	Within \(\BigO{\log \min \Set{d, m}}\) rounds, every amoebot \(p_i\) can compare \(i\) to \(d\).
	The procedure can run simultaneously on any set of edge-disjoint chains with length \(\leq m-1\).
\end{lemma}
\begin{proof}
	Consider some chain \(C = (p_0, \ldots, p_{m-1})\) and let \(d \in \Nats_0\) be stored in some binary counter.
	First, the amoebots find the most significant bit of \(d\), which takes only a constant number of rounds.
	In the degenerate case \(d = 0\), the amoebot at the start of the counter sends a beep on a global circuit and each amoebot \(p_i\) locally compares \(i\) to \(0\), which it can do by checking the existence of its predecessor (only \(p_0\) has no predecessor).
	This takes a constant number of rounds.
	
	For \(d > 0\), the amoebots then run the \PASC{} algorithm on \(C\), using \(p_0\) as the start point, which allows each amoebot \(p_i\) to obtain the bits of \(i\) as a bit stream by Lemma~\ref{lem:pasc}.
	Simultaneously, they transmit the bits of \(d\) on a global circuit by moving a marker along the counter on which \(d\) is stored and letting it beep on the global circuit whenever its current bit is \(1\).
	The two procedures are synchronized such that after each \PASC{} iteration, one bit of \(d\) is transmitted.
	Thus, each amoebot \(p_i\) receives two bit streams, one for \(i\) and one for \(d\).
	By Lemma~\ref{lem:binary_stream_operations}, this already allows \(p_i\) to compare \(i\) to \(d\), if we let the procedure run for \(\Floor{\log \max \Set{d, m-1}} + 1\) iterations.
	
	If \(d\) and \(m-1\) have the same number of bits, we are done already.
	Otherwise, either the \PASC{} algorithm or the traversal of \(d\) will finish first.
	The amoebots can recognize all three cases by establishing the global circuit for two additional rounds per iteration and letting the amoebots involved in the unfinished procedures beep, using one round for the \PASC{} algorithm and the other round for the traversal of \(d\).
	Now, if the \PASC{} algorithm finishes first but there is still at least one non-zero bit of \(d\) left, then we must have \(d > m-1\), so the comparison result is simply \(i < d\) for all \(p_i\).
	Conversely, if the traversal of \(d\) finishes first, the amoebots establish a circuit along \(C\) by letting all \(p_i\) connect their predecessor and successor except the ones whose \emph{current} comparison result is \(i = d\) (note that there may be more than one such amoebot).
	The closest such amoebot to \(p_0\) on the chain will be the one with \(i = d\); it has already received all non-zero bits of \(i\) because \(i\) has just as many bits as \(d\).
	The start of the chain, \(p_0\), now sends a beep towards its successor, which will reach all amoebots \(p_i\) with \(i \leq d\).
	Thereby, all amoebots \(p_i\) on the chain know whether \(i \leq d\) (beep received), \(i = d\) (beep received and comparison is equal), or \(i > d\) (no beep received).
	
	In both cases, we only require a constant number of rounds after finishing the first procedure, implying the runtime of \(\BigO{\log \min \Set{d, m}}\) rounds.
	Finally, consider a set of chains with maximum length \(m-1\) where no two chains share an edge.
	Because no edge is shared and the \PASC{} algorithm only uses edges on its chain, all chains can run the \PASC{} algorithm simultaneously without interference.
	The same holds for the chain circuits used for the case \(d < m-1\).
	In the synchronization rounds, a beep is now sent on the global circuit whenever \emph{any} of the \PASC{} executions is not finished yet.
	If all chains require the same number of \PASC{} iterations, there is no difference to the case with a single chain.
	If any chain finishes its \PASC{} execution earlier, it can already finish its own procedure with the result \(i < d\) for all its amoebots \(p_i\) without influencing the other chains.
\end{proof}


\section{A Simple Lower Bound}
\label{sec:lower_bound}

We first show a lower bound that demonstrates a central difficulty arising in the shape containment problem.
For a simple example shape (see Fig.~\ref{fig:lower_bound_shape}), we show that even if the maximum scale is known, identifying all valid placements of the target shape can require \(\BigOmega{\sqrt{n}}\) rounds due to communication bottlenecks.

\begin{theorem}
	\label{theo:lower_bound}
	There exists a shape \(S\) such that for any choice of origin and every amoebot algorithm \(\mathcal{A}\) that terminates after \(\LittleO{\sqrt{n}}\) rounds, there exists an amoebot structure \(A\) for which the algorithm does not compute \(\VP(\kmax(S, A) \cdot S, A)\), even if \(\kmax\) is known.
\end{theorem}
\begin{proof}
    \begin{figure}
        \centering
        \includegraphics[width=0.35\linewidth]{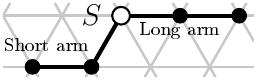}
        \caption{An example shape for which the valid placement search is bounded below by \(\BigOmega{\sqrt{n}}\).}
        \label{fig:lower_bound_shape}
    \end{figure}

	We use the shape \(S\) with a long arm and a short arm connected by a diagonal edge, as depicted in Fig.~\ref{fig:lower_bound_shape}.
	Let \(\mathcal{A}\) be an amoebot algorithm that terminates in \(\LittleO{\sqrt{n}}\) rounds.
	For every \(k \in \Nats\), we will construct a set \(A_k\) of amoebot structures such that \(\kmax(S, A) = k\) for all \(A \in A_k\) and only one rotation matches at this scale.
	Let \(k \in \Nats\) be arbitrary, then we construct \(A_k\) as follows (see Fig.~\ref{fig:lower_bound_structure} for reference):
	
	First, we place a parallelogram of width \(2k\) and height \(k-1\) with its lower left corner at the origin and call this the \emph{first block}.
	The first block contains \((2k + 1) k = 2k^2 + k\) amoebots and is shared by all \(A \in A_k\).
	Let \(p_0, \ldots, p_{k-1}\) be the nodes occupied by the left side of the parallelogram, ordered from bottom to top.
	Next, we place a second parallelogram with width and height \(k-1\) such that its right side extends the first block's left side below the origin.
	This second block contains \(k^2\) amoebots and is also the same for all structures.
	It is only connected to the first block by a single edge, \(e\).
	Let \(q_0, \ldots, q_{k-1}\) be the nodes one step to the left of the second block, again ordered from bottom to top.

	\begin{figure}
		\centering
		\includegraphics[width=0.9\linewidth]{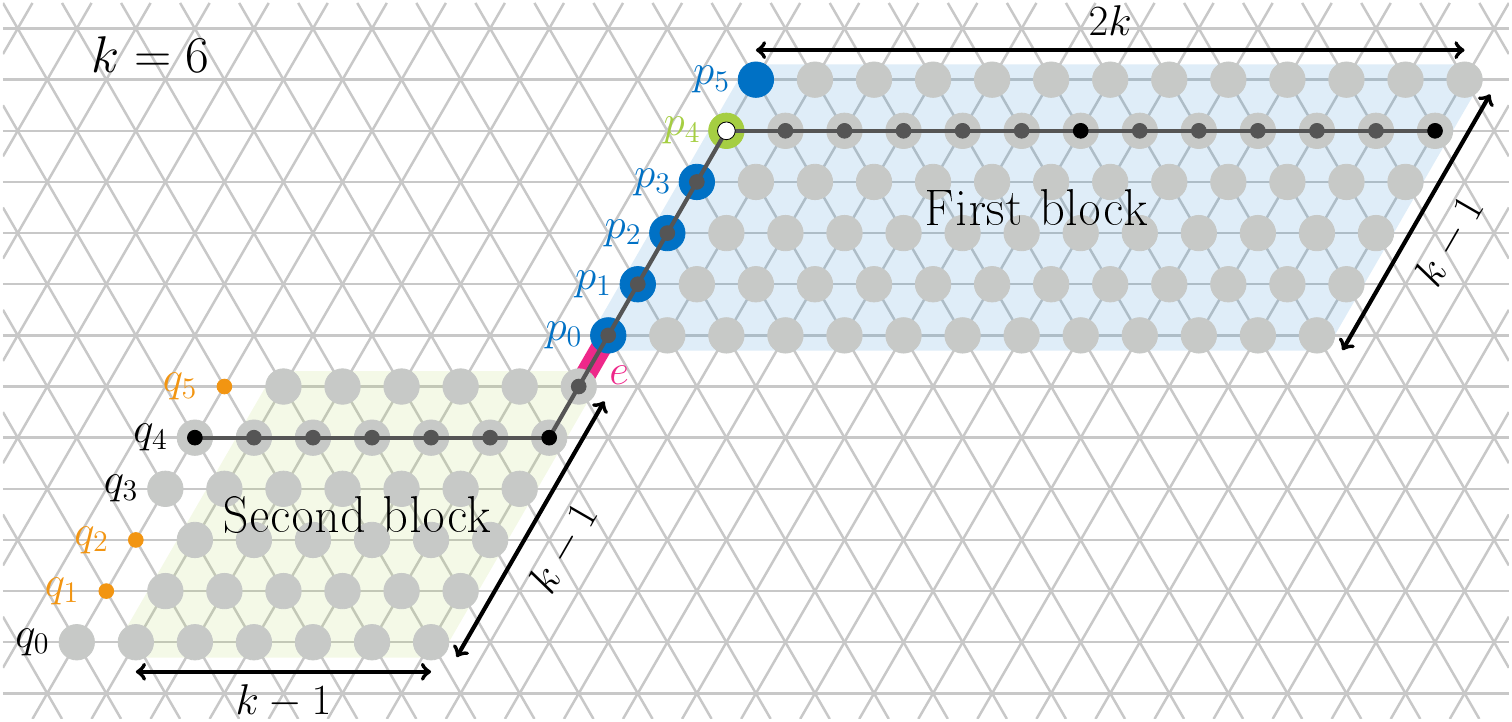}
		\caption{Overview of the amoebot system construction for scale \(k = 6\).
			The first block is shaded blue and the second block is shaded green.
			The nodes \(q_i\) that are not contained in the structure are colored orange.
			Amoebot \(p_4\) is a valid placement of \(k \cdot S\) because \(q_4\) is part of the structure.}
		\label{fig:lower_bound_structure}
	\end{figure}

	We define \(A_k\) as the set of amoebot structures that consist of these two blocks and \(m\) additional amoebots on the positions \(q_0, \ldots, q_{k-1}\), where \(1 \leq m \leq k\).
	Thus, \(A_k\) contains \(2^k - 1\) distinct structures.
	Now, consider placements of \(S\) with maximum scale in any structure \(A \in A_k\).
	For \(m = k\), there are exactly \(k\) valid placements at scale \(k\), represented by the amoebots \(p_0, \ldots, p_{k-1}\).
	The longest continuous lines of amoebots in \(A\) have length \(2k\) and form the first block.
	In every valid placement, the longer arm of \(k \cdot S\) must occupy one of these lines, so no larger scales or other rotations are possible.
	If \(q_i\) is not occupied for some \(0 \leq i \leq k-1\), then \(p_i\) is not a valid placement because the end of the shorter arm of \(k \cdot S\) would be placed on \(q_i\).
	At least one \(q_i\) is always occupied, so the maximum scale of \(S\) is \(k\) for every \(A \in A_k\).
	Observe that every structure \(A \in A_k\) has a unique configuration of valid placements of \(k \cdot S\):
	\(p \in \VP(k \cdot S, A)\) if and only if \(p = p_i\) and \(q_i \in A\) for some \(0 \leq i \leq k-1\).
	
	Next, consider the size of the structures in \(A_k\).
	The maximum number of amoebots is \(2k^2 + k + k^2 + k = 3k^2 + 2k\), obtained for \(m = k\).
	This means we have \(n \leq 3k^2 + 2k \leq 4k^2\) for large enough \(k\), \ie{}, \(k \geq \sqrt{n}/2\) for all \(k \geq 2\) and all \(A \in A_k\).
	
	Let \(A \in A_k\) be arbitrary and consider the final states of \(p_0, \ldots, p_{k-1}\) after \(\mathcal{A}\) has been executed on \(A\).
	Each amoebot must be categorized as either a valid or an invalid placement of \(k \cdot S\).
	We can assume that this categorization is independent of any randomized decisions because otherwise, there would be a non-zero probability of false categorizations.
	Thus, the final state depends only on the structure \(A\) itself.
	Recall that structures in \(A_k\) only differ in the positions \(q_0, \ldots, q_{k-1}\) and every path between \(q_i\) (or an occupied neighbor) and \(p_i\) must traverse the single edge \(e\) connecting the two blocks.
	We can assume that all communication happens via circuits (see Sec.~\ref{subsec:circuit_extension}).
	Since the first block is the same in all structures, the final states of \(p_0, \ldots, p_{k-1}\) only depend on the sequence of signals sent from the second block to the first block through \(e\).
	In order to compute the correct set of valid placements, each amoebot structure in \(A_k\) therefore has to produce a unique sequence of signals:
	If for any two configurations, the same sequence of signals is sent through \(e\), the final states of \(p_0, \ldots, p_{k-1}\) will be identical, so at least one will be categorized incorrectly.
	
	Let \(c\) be the number of pins used by \(\mathcal{A}\).
	Then, the number of different signals that can be sent via one edge in one round is \(2^c = \BigO{1}\) and the number of signal sequences that can be sent in \(r\) rounds is \(2^{rc}\).
	Therefore, to produce at least \(2^k - 1\) different sequences of signals, we require \(r = \BigOmega{k/c} = \BigOmega{\sqrt{n}}\) rounds.
	By the assumption that \(\mathcal{A}\) terminates after \(\LittleO{\sqrt{n}}\) rounds, \(\mathcal{A}\) will produce at least one false result for sufficiently large \(k\).
	
	It remains to be shown that the same arguments hold for all equivalent versions of \(S\) that contain the origin.
	If the origin is placed on another node of the longer arm, the valid placement candidates \(p_0, \ldots, p_{k-1}\) are simply shifted to the right by \(k\) or \(2k\) steps, respectively, everything else remains the same.
	If the origin is placed on the shorter arm of the shape, we switch the roles of the first and the second block.
	We place amoebots on all positions \(q_0, \ldots, q_{k-1}\) and use the right side of the first block as the controlling positions instead.
	The number and size of the resulting amoebot structures remain the same, so the same arguments hold as before.
\end{proof}


\section{Helper Procedures}
\label{sec:helpers}

In this section, we introduce the basic primitives we will use to construct shapes for which our valid placement search procedures get below the lower bound.


\subsection{Scale Factor Search}
\label{subsec:scale_factor_search}

As outlined earlier, our shape containment algorithms consist of two search procedures.
The first is a \emph{scale factor search} that determines which scales have to be checked in order to find the maximum scale, and the second procedure is a \emph{valid placement search} that identifies all valid placements of \(k \cdot S^{(r)}\) for all \(r \in \Set{0, \ldots, 5}\) and the scale \(k\), given in a binary counter.

Consider some shape \(S\) and an amoebot structure \(A\) with a binary counter that stores an upper bound \(K \geq \kmax(S, A)\).
The simple \emph{linear search} procedure runs valid placement checks for the scales \(K, K-1, \ldots, 1\) and accepts when the first valid placement is found.
If no placement is found in any iteration, we have \(\kmax = 0\).

\begin{lemma}
	\label{lem:linear_search}
	Let \(S\) be a shape and \(A\) an amoebot structure with a binary counter storing an upper bound \(K \geq \kmax(S, A)\).
	Given a valid placement search procedure for \(S\), the amoebots compute \(\kmax(S, A)\) in at most \(K\) iterations, running the placement search for scales \(K, K-1, \ldots, \kmax\) and with constant overhead per iteration.
\end{lemma}
\begin{proof}
	Let \(S\) be a shape, \(A\) an amoebot structure storing \(K \geq \kmax(S, A)\) in a binary counter and let a valid placement search procedure for \(S\) be given.
	In the first iteration, the amoebots run this procedure to compute \(\VP(K \cdot S^{(r)}, A)\) for all \(r \in \Set{0, \ldots, 5}\).
	If any of these sets is not empty, we have \(\kmax(S, A) = K\) and the procedure terminates.
	Otherwise, by Lemma~\ref{lem:binary_operations}, the amoebots can compute \(K - 1\) and compare it to \(0\) in a constant number of rounds.
	If it is \(0\), we have \(\kmax = 0\), and otherwise, we repeat the above steps for \(K - 1\).
	Since \(K\) is reduced by \(1\) in each step, this takes at most \(K\) iterations overall.
\end{proof}

When using the linear search method, finding a small upper bound \(K\) is essential for reducing the runtime.
However, some shapes permit a faster search method based on an inclusion relation between different scales.

\begin{definition}
	\label{def:scale_shapes}
	We call a shape \(S\) \emph{self-contained} if for all scales \(k < k'\), there exist a translation \(t \in \Reals^2\) and a rotation \(r \in \Set{0, \ldots, 5}\) such that \(k \cdot S^{(r)} + t \subseteq k' \cdot S\).
\end{definition}
For self-contained shapes, finding no valid placements at scale \(k\) immediately implies \(\kmax(S, A) < k\), which allows us to apply a \emph{binary search}.

\begin{lemma}
	\label{lem:binary_search}
	Let \(S\) be a self-contained shape and let \(A\) be an amoebot structure with a binary counter large enough to store \(\kmax = \kmax(S, A)\).
	Given a valid placement search procedure for \(S\), the amoebots can compute \(\kmax\) within \(\BigO{\log \kmax}\) iterations such that each iteration runs the valid placement search once for some scale \(k \leq 2 \cdot \kmax\) and has constant overhead.
\end{lemma}

To prove Lemma~\ref{lem:binary_search}, we first show the following result, which guarantees the existence of valid placements for self-contained shapes if there is already a valid placement at a larger scale.

\begin{lemma}
	\label{lem:grid_translation}
	Let \(S\) and \(S'\) be arbitrary shapes for which there is a translation \(t \in \Reals^2\) such that \(S + t \subseteq S'\).
	Then, every \(t' \in \Veqt\) with minimal Euclidean distance to \(t\) satisfies \(S + t' \subseteq S'\).
\end{lemma}
\begin{proof}
	Consider arbitrary shapes \(S, S'\) with a translation \(t \in \Reals^2\) such that \(S + t \subseteq S'\).
	Let \(t' \in \Veqt\) be a grid node with minimum Euclidean distance to \(t\) and suppose \(t \neq t'\) (otherwise we are done).
	The distance between \(t\) and \(t'\) is at most \(\sqrt{3}/3\) because this is the distance between all of a grid face's corners and its center.
	Now, consider any grid node \(p \in S\).
	Since \(p + t\) is contained in \(S'\), \(p + t\) must be located on an edge or face of \(S'\).
	In each case, \(p + t'\) is a closest node to \(p + t\) and this node must be occupied by \(S'\) since it must belong to that edge or face.
	
	Next, consider some edge \(e \subseteq S\) and let its two end points be \(p\) and \(q\).
	If \(p + t\) and \(q + t\) lie on edges parallel to \(e\) in \(S'\), then \(e + t'\) clearly coincides with one of these edges.
	If \(p + t\) and \(q + t\) both lie on edges not parallel to \(e\), \(S'\) must contain the parallelogram spanned by those edges and \(e + t'\) will lie on one of the sides of the parallelogram.
	Otherwise, \(p + t\) and \(q + t\) must lie in two faces of \(S'\) which have the same orientation and share one corner while \(e + t\) crosses the face between them.
	Since \(e + t'\) will lie on a side of one of these faces and \(S'\) must contain all of them, \(e + t'\) will be contained as well.
	
	
	Finally, let \(f \subseteq S\) be some face and let \(p\) be its center.
	Observe that the minimal distance to the center of a face in \(S'\) with similar orientation that does not intersect \(f\) is the face height \(\sqrt{3}/2\), which is greater than \(\sqrt{3}/3\).
	Thus, \(f + t\) must already intersect the face \(f + t'\), which therefore has to be contained in \(S'\).
\end{proof}

In particular, Lemma~\ref{lem:grid_translation} implies that a shape \(S\) is self-contained if and only if for all scales \(k < k'\), there are a rotation \(r\) and a \emph{grid node} \(t \in \Veqt\) such that \(k \cdot S^{(r)} + t \subseteq k' \cdot S\).

\begin{proof}[Proof of Lemma~\ref{lem:binary_search}]
	Let \(S\) be a self-contained shape.
	Consider an amoebot structure \(A\) with a binary counter that can store \(\kmax = \kmax(S, A)\) and suppose there is a valid placement search procedure for \(S\).
	We run a \emph{binary search} as follows:
	
	First, the amoebots run the valid placement search for scale \(k = 1\) and all valid placements (for any rotation) beep on a global circuit.
	If no beep is sent, the maximum scale must be \(0\) and the procedure terminates.
	Otherwise, the amoebots compute \(k \gets 2 \cdot k\) on the binary counter and run the valid placement search again for the new scale.
	They repeat this until no valid placement is found for the current scale \(k\), at which point an upper bound \(U := k > \kmax\) has been found.
	Observe that \(U \leq 2 \cdot \kmax\), so the counter requires at most one more bit to store \(U\) than for \(\kmax\), which can be handled by the last amoebot in the counter by simulating its successor.
	
	We now maintain the upper bound \(U > \kmax\) and the lower bound \(L := 1 \leq \kmax\) as a loop invariant during the following binary search.
	In each iteration, the amoebots compute \(k = \Floor{(L + U) / 2}\) and run the valid placement search procedure for scale \(k\).
	If a valid placement is found, we update \(L \gets k\), otherwise we update \(U \gets k\).
	We repeat this until \(U = L + 1\), at which point we have \(L = \kmax\).
	Each of these two phases takes \(\BigO{\log \kmax}\) iterations, as is commonly known for binary search algorithms, and we run only one valid placement search in each iteration.
	
	To show the correctness, let \(k\) be some scale factor.
	If there is a valid placement for scale \(k\), then we clearly have \(\kmax(S, A) \geq k\).
	If there are no valid placements for scale \(k\), consider any scale \(k' > k\) and a translation \(t \in \Reals^2\) and rotation \(r \in \Set{0, \ldots, 5}\) such that \(k \cdot S^{(r)} + t \subseteq k' \cdot S\), which exist because \(S\) is self-contained.
	By Lemma~\ref{lem:grid_translation}, \(t\) can always be chosen as a grid node position so that \(k \cdot S^{(r)} + t\) aligns with the grid.
	Then, for any \(p \in \VP(k' \cdot S, A)\), the amoebot at location \(p + t\) is a valid placement of \(k \cdot S^{(r)}\) since \(V(k \cdot S^{(r)} + p + t) \subseteq V(k' \cdot S + p) \subseteq A\).
	By our assumption that there are no valid placements for scale \(k\), we have \(\VP(k' \cdot S, A) = \emptyset\) for any choice of \(k'\), implying \(\kmax(S, A) < k\).
	Therefore, the invariants \(U > \kmax\) and \(L \leq \kmax\) are established in the first phase and are maintained during the binary search in the second phase.
	As a consequence, \(L = \kmax\) holds when \(U = L + 1\) is reached.
	Additionally, since \(k \leq U\) and \(U \leq 2 \cdot \kmax\) for every checked scale \(k\), the valid placement search is only executed for scales at most \(2 \cdot \kmax\).
\end{proof}


\subsection{Primitive Shapes}
\label{subsec:primitive_shapes}

\begin{definition}
	\label{def:line_shape}
	A \emph{line shape} \(\Line(d, \ell)\) is a shape consisting of \(\ell \in \Nats_0\) consecutive edges extending in direction \(d\) from the origin.
	For \(\ell = 0\), the shape contains only the origin point.
\end{definition}
\begin{definition}
	\label{def:triangle_shape}
	Let \(\Tri(d, 1)\) be the shape consisting of the triangular face spanned by the unit vectors \(\UVec{d}\) and \(\UVec{d'}\), where \(d'\) is obtained from \(d\) by one \(60^\circ\) counter-clockwise rotation.
	We define general \emph{triangle shapes} as \(\Tri(d, \ell) := \ell \cdot \Tri(d, 1)\) for \(\ell \in \Nats_{> 1}\) and call \(\ell\) the \emph{side length} or \emph{size} of \(\Tri(d, \ell)\).
\end{definition}
Lines and triangles are important primitive shapes which we will use to construct more complex shapes.
In this subsection, we introduce placement search procedures allowing amoebots to identify valid placements of these shapes when their size is given in a binary counter.
The procedures rely heavily on the \PASC{} algorithm combined with binary operations on bit streams.
A simple and natural way to establish the required chains is using \emph{segments}:

\begin{definition}
	\label{def:segment}
	Let \(W \in \Set{X, Y, Z}\) be a grid axis.
	A \emph{(\(W\))-segment} is a connected set of nodes on a line parallel to \(W\).
	Let \(C \subseteq \Veqt\), then a \emph{maximal \(W\)-segment of \(C\)} is a finite \(W\)-segment \(M \subseteq C\) that cannot be extended with nodes from \(C\) on either end.
	The \emph{length} of a finite segment \(M\) is \(\Abs{M} - 1\).
\end{definition}
For example, chains on maximal segments of the amoebot structure \(A\) can be constructed easily once a direction has been agreed upon:
All amoebots on a segment identify their chain predecessor and successor by checking the existence of neighbors on the direction's axis.
The start and end points of the segment are the unique amoebots lacking a neighbor in one or both directions.

Our placement search procedure for lines essentially runs the \PASC{} algorithm to measure the length of amoebot segments and compares them to the given scale.
We construct the procedure in several steps.
First, running the \PASC{} algorithm on maximal amoebot segments allows the amoebots to compute their distance to a boundary:

\begin{lemma}
    \label{lem:max_line_length}
    Let \(A\) be an amoebot structure and \(d \in \Directions\) be a cardinal direction known by the amoebots.
    Within \(\BigO{\log n}\) rounds, each amoebot \(p \in A\) can compute its own distance to the nearest boundary in direction \(d\) as a sequence of bits.
\end{lemma}
Observe that this boundary distance is the largest \(\ell \in \Nats_0\) such that \(p \in \VP(\Line(d, \ell), A)\).
If a desired line length is given, the amoebots can use this procedure to determine the valid placements of the line:
\begin{lemma}
    \label{lem:line_detection}
    Let \(L = \Line(d, \ell)\) be a line shape and let \(A\) be an amoebot structure that knows \(d\) and stores \(\ell\) in some binary counter.
    Within \(\BigO{\log \min \Set{\ell, n}}\) rounds, the amoebots can compute \(\VP(L, A)\).
\end{lemma}
\begin{proof}[Proof of Lemmas \ref{lem:max_line_length} and \ref{lem:line_detection}]
	Let \(A\) be an amoebot structure and \(d \in \Directions\) a direction known by the amoebots.
	First, the amoebots establish chains along all maximal segments in the opposite direction of \(d\), such that on each segment, the amoebot furthest in direction \(d\) is the start of the chain.
	This can be done in one round since each amoebot simply chooses its neighbor in direction \(d\) as its predecessor and the neighbor in the opposite direction as its successor.
	Next, the amoebots run the \PASC{} algorithm on all segments simultaneously, synchronized using a global circuit.
	This allows each amoebot to compute the distance to its segment's end point in direction \(d\) as a bit sequence by Lemma~\ref{lem:pasc}.
	Because the length of each segment is bounded by \(n\), the \PASC{} algorithm terminates within \(\BigO{\log n}\) rounds.
	
	Now, suppose a length \(\ell\) is stored in some binary counter in \(A\).
	We modify the procedure such that in each iteration of the \PASC{} algorithm, we transmit one bit of \(\ell\) on the global circuit.
	By Lemma~\ref{lem:pasc_cutoff}, this allows each amoebot to compare its distance to the boundary in direction \(d\) to \(\ell\), since the segments are disjoint (and therefore edge-disjoint in particular).
	We have \(p \in \VP(\Line(d, \ell), A)\) if and only if the distance of amoebot \(p\) to the boundary in direction \(d\) is at least \(\ell\).
	Thus, each amoebot can immediately decide whether it is in \(\VP(\Line(d, \ell), A)\) after the comparison, which takes \(\BigO{\log \min \Set{\ell, n}}\) rounds.
\end{proof}

Next, consider the problem of finding all longest segments in the amoebot structure \(A\).
This is equivalent to solving the shape containment problem for any base shape \(\Line(d, 1)\) with \(d \in \Directions\).

\begin{lemma}
    \label{lem:longest_segments}
    For any direction \(d \in \Directions\), the shape containment problem for the line shape \(\Line(d, 1)\) can be solved in \(\BigO{\log k}\) rounds, where \(k = \kmax(\Line(d, 1), A)\).
\end{lemma}
\begin{proof}
	Consider some amoebot structure \(A\) and let \(m\) be the maximum length of a segment in \(A\).
	The amoebots first establish chains along all maximal \(X\)-, \(Y\)- and \(Z\)-segments and run the \PASC{} algorithm on them to compute their lengths.
	On each segment, the end point transmits the received bits on a circuit spanning the whole segment so that it can be stored in the segment itself, using it as a counter.
	This works simultaneously because segments belonging to the same axis are disjoint and because each amoebot stores at most three bits (one for each axis).
	We use a global circuit for synchronization and let each segment beep as long as it has not finished computing its length.
	Any segment that is already finished but receives a beep on the global circuit marks itself as retired since it cannot have maximal length.
	At the end of this step, the lengths of all non-retired segments are stored on the segments and share the same number of bits, which is equal to \(\Floor{\log m} + 1\).
	Next, each segment places a marker on its highest-value bit and moves it backwards along the chain, one step per iteration.
	For each bit, the segment beeps on a global circuit if the bit's value is \(1\).
	If the bit's value is \(0\) but a beep was received on the global circuit, the segment retires since the other segment that sent the beep must have a greater length.
	This is true because at this point, all previous (higher-value) bits of the two lengths must have been equal, so the current bit is the first (and therefore highest value) position where the two numbers differ.
	At the end of the procedure, all segments with length less than \(m\) have retired.
	Because the segments of length \(m\) have not retired in any iteration and since \(\kmax(\Line(d, 1), A) = m\), the algorithm solves the containment problem for \(\Line(d, 1)\).
	The runtime follows directly from the runtime of the \PASC{} algorithm (Lemma~\ref{lem:pasc}).
\end{proof}

Using Lemmas~\ref{lem:line_detection} and \ref{lem:longest_segments}, we obtain a simple solution for lines of arbitrary base lengths:

\begin{corollary}
    \label{cor:longest_lines}
    For any direction \(d \in \Directions\) and length \(\ell \in \Nats\), the shape containment problem for the line shape \(\Line(d, \ell)\) can be solved in \(\BigO{\log m}\) rounds, where \(m\) is the length of a longest segment in \(A\).
\end{corollary}
\begin{proof}
	By Lemma~\ref{lem:longest_segments}, the amoebots can find the maximal segment length \(m\) in \(A\) and write it into binary counters within \(\BigO{\log m}\) rounds, establishing counters on the longest segments.
	After that, on each counter storing \(m\), they can compute \(k := \Floor{m / \ell}\) in \(\BigO{\log m}\) rounds by Lemma~\ref{lem:binary_operations} and since \(\ell\) is a constant with a known binary representation.
	The maximum scale for \(\Line(d, \ell)\) is \(k\) since for \(k + 1\), we would require segments of length \((k + 1) \cdot \ell > m\) in \(A\).
	Finally, using Lemma~\ref{lem:line_detection}, the amoebots find all valid placements of \(\Line(d, k \cdot \ell)\) in \(\BigO{\log \min \Set{m, n}} = \BigO{\log m}\) rounds.
	The procedure can be repeated a constant number of times for the other rotations of the line.
\end{proof}

Moving on, our \emph{triangle primitive} constructs valid placements of triangles.

\begin{lemma}
    \label{lem:triangle_primitive}
    Let \(T = \Tri(d, \ell)\) be a triangle shape and let \(A\) be an amoebot structure that knows \(d\) and stores \(\ell\) in some binary counter.
    The amoebots can compute \(\VP(T, A)\) within \(\BigO{\log \min\Set{\ell, n}}\) rounds.
\end{lemma}
The procedure runs the line primitive to find valid placements of lines and then applies techniques from the following subsections to transform and combine them into valid placements of triangles.
We defer the proof of Lemma~\ref{lem:triangle_primitive} until the relevant ideas have been explained (see Sec.~\ref{subsec:valid_placement_search}) since the approach will be useful for more shapes than triangles.

Observe that for any two shapes \(S, S'\) with \(S \subseteq S'\), \(\kmax(S, A)\) is an upper bound for \(\kmax(S', A)\).
Thus, the maximal scale of an edge \(\Line(d, 1)\) or face \(\Tri(d, 1)\) is a natural upper bound on the maximum scale of any shape containing an edge or face, respectively.
For this reason, any longest segment in the amoebot structure provides sufficient memory to store the scale values we have to consider.
To use this fact, we will establish binary counters on \emph{all} maximal amoebot segments (on all axes) and use them simultaneously, deactivating the ones whose memory is exceeded at any point.
Using Lemma~\ref{lem:triangle_primitive} therefore allows us not only to solve the shape containment problem for triangles but also to determine an upper bound on the scale of shapes that contain a triangle.

\begin{corollary}
    \label{cor:largest_triangles}
    Let \(T = \Tri(d, 1)\), let \(A\) be some amoebot structure, and \(k = \kmax(T, A)\).
    Within \(\BigO{\log^2 k}\) rounds, the amoebots can solve the shape containment problem for \(T\) and store \(k\) in some binary counter.
\end{corollary}
\begin{proof}
	Because triangles are convex and therefore self-contained, we can apply a binary search for the maximum scale factor \(k = \kmax(T, A)\), which requires \(\BigO{\log k}\) iterations and only checks scales \(\ell \leq 2 \cdot k\) by Lemma~\ref{lem:binary_search}.
	To provide a binary counter of sufficient size, the amoebots can establish binary counters on \emph{all} maximal segments of \(A\) and use them all simultaneously, deactivating the counters whose memory is exceeded during some operation.
	At least one of these will have sufficient size to store \(k\) because \(T\) contains an edge, so \(k\) is bounded by \(\kmax(\Line(d, 1), A)\).
	By Lemma~\ref{lem:triangle_primitive}, the valid placement search for a triangle of size \(\ell\) only requires \(\BigO{\log \min\Set{\ell, n}}\) rounds, which already proves the runtime.
	The maximum scale is still stored on the binary counters after the scale factor search.
\end{proof}


\subsection{Stretched Shapes}
\label{subsec:stretched_shapes}

With the ability to compute valid placements of some basic shapes, we now consider operations on shapes that allow us to quickly determine the valid placements of a transformed shape.
The first, simple operation is the \emph{union} of shapes.
Given the valid placements \(C_1 = \VP(S_1, A)\) and \(C_2 = \VP(S_2, A)\) of two shapes \(S_1\) and \(S_2\), the amoebots in \(A\) can find the valid placements of \(S' = S_1 \cup S_2\) in a single round:
Due to the relation \(\VP(S', A) = C_1 \cap C_2\), each amoebot locally decides whether it is a valid placement of both shapes.

Next, we consider the \emph{Minkowski sum} of a shape with a line.
\begin{definition}
	\label{def:minkowski_sum}
	Let \(S_1, S_2\) be two shapes, then their \emph{Minkowski sum} is defined as
	\[
		S_1 \oplus S_2 := \{p_1 + p_2 \mid p_1 \in S_1, p_2 \in S_2\}.
	\]
\end{definition}
The resulting subset of \(\Reals^2\) is a valid shape and if both shapes contain the origin, then their sum also contains the origin.
Observe that for any shape \(S\) and any line \(\Line(d, \ell)\), we have
\[
	V(S \oplus \Line(d, \ell)) = \bigcup\limits_{i = 0}^\ell V(S + i \cdot \UVec{d}).
\]
Let \(S' = S \oplus \Line(d, \ell)\), then \(S'\) is a "stretched" version of \(S\).
Consider the valid placements \(C = \VP(S, A)\) of \(S\) in \(A\).
Now, if \(p \in A \setminus C\), then neither \(p\) nor the \(\ell\) positions in the opposite direction of \(d\) relative to \(p\) are valid placements of \(S'\) because placing \(S'\) at any of these positions would require a copy of \(S\) placed on \(p\).
Using the \PASC{} algorithm on the segment that starts at \(p\) and extends in the opposite direction of \(d\), we can therefore eliminate placement candidates of \(S'\).

\begin{lemma}
	\label{lem:stretched_placements}
	Let \(S\) be an arbitrary shape, \(L = \Line(d, \ell)\) a line and \(A\) an amoebot structure storing a scale \(k\) in some binary counter.
	Given \(d, \ell\) and the set \(C = \VP(k \cdot S, A)\), the amoebots can compute \(\VP(k \cdot (S \oplus L), A)\) within \(\BigO{\log \min\Set{k \cdot \ell, n}}\) rounds.
\end{lemma}
\begin{proof}
	Let \(S\) and \(L = \Line(d, \ell)\) be arbitrary and consider an amoebot structure \(A\) storing \(k\) in a binary counter.
	Suppose every amoebot in \(A\) knows whether it is part of \(C = \VP(k \cdot S, A)\) and let \(S' = S \oplus L\).
	First, observe that the node set covered by \(k \cdot S'\) is the union of \(k \cdot \ell + 1\) copies of \(V(k \cdot S)\):
	\begin{equation}
		V(k \cdot S')
		= V((k \cdot S) \oplus \Line(d, k \cdot \ell))
		= \bigcup\limits_{i = 0}^{k \cdot \ell} V(k \cdot S + i \cdot \UVec{d})
		\label{eq:stretched_node_set}
	\end{equation}
	Additionally, since \(S\) contains the origin, we have \(k \cdot L \subseteq k \cdot S'\), implying \(\VP(k \cdot S', A) \subseteq \VP(k \cdot L, A)\).
	
	Let \(C' = A\) be the initial set of placement candidates for \(k \cdot S'\).
	The amoebots first run the line placement search for \(k \cdot L\), identifying \(\VP(k \cdot L, A)\) within \(\BigO{\log \min \Set{k \cdot \ell, n}}\) rounds by Lemma~\ref{lem:line_detection}.
	Since \(\ell\) is known by the amoebots, they can compute \(k \cdot \ell\) in constant time.
	The invalid placements of \(k \cdot L\) remove themselves from \(C'\) since they cannot be valid placements of \(k \cdot S'\).
	
	Now, consider some \(q \in A \setminus C\), then the node set \(M(q) = V(k \cdot S + q)\) is not fully covered by \(A\).
	However, for every \(0 \leq i \leq k \cdot \ell\), we have \(M(q) \subseteq V(k \cdot S' + q - i \cdot \UVec{d})\) due to \eqref{eq:stretched_node_set}.
	Thus, neither \(q\) nor any position \(q - i \cdot \UVec{d}\) for \(1 \leq i \leq k \cdot \ell\) is a valid placement of \(k \cdot S'\).
	Let \(Q(q) = \SetBar{q - i \cdot \UVec{d}}{0 \leq i \leq k \cdot \ell}\) be the set of these invalid placements and observe that \(Q(q)\) is a (not necessarily occupied) \(W\)-segment of length \(k \cdot \ell\) with one end point at \(q\), where \(W\) is the grid axis parallel to \(d\).
	Further, let \(N(q) \subseteq A\) be the maximal \(W\)-segment of \(A\) that contains \(q\).
	If \(q\) is the only amoebot in \(A \setminus C\) on \(N(q)\), then by Lemma~\ref{lem:pasc_cutoff}, the amoebots \(Q(q) \cap N(q)\) can identify themselves within \(\BigO{\log \min \Set{k \cdot \ell, n}}\) rounds, using \(q\) as the start of a chain on \(N(q)\) that extends in the opposite direction of \(d\) and transmitting \(k \cdot \ell\) on a global circuit.
	Those amoebots with distance at most \(k \cdot \ell\) to \(q\) on this chain are the ones in \(N(q) \cap Q(q)\) and they remove themselves from \(C'\).
	
	If there are multiple invalid placements of \(k \cdot S\) on \(N(q)\), the amoebots establish one such chain for each, extending in the opposite direction of \(d\) until the next invalid placement or the boundary of the structure.
	Now, if \(Q(q) \cap N(q)\) contains some amoebot \(q' = q - j \cdot \UVec{d} \in A \setminus C\), the remaining amoebots \(q - i \cdot \UVec{d}\) for \(j < i \leq k \cdot \ell\) are contained in \(Q(q')\) and will be identified (on \(N(q)\)).
	Because the chains on \(N(q)\) are disjoint by construction and the maximal \(W\)-segments of \(A\) are also disjoint, all amoebots in any set \(Q(q)\) with \(q \in A \setminus C\) can be determined within \(\BigO{\log \min \Set{k \cdot \ell, n}}\) rounds by Lemma~\ref{lem:pasc_cutoff}.
	
	All amoebots that are removed from \(C'\) by these two steps are invalid placements of \(k \cdot S'\).
	To show that \emph{all} invalid placements are removed, consider some \(q \in A \setminus \VP(k \cdot S', A)\).
	If \(q \notin \VP(k \cdot L, A)\), \(q\) will be removed by the line check.
	Otherwise, there must be a position \(0 \leq i \leq k \cdot \ell\) such that \(q + i \cdot \UVec{d} \notin C\) because otherwise, \(q\) would be a valid placement.
	Let \(q' = q + i \cdot \UVec{d}\) be such an amoebot with minimal \(i\), then \(q \in Q(q')\) and \(q \in N(q')\) (because \(q \in \VP(k \cdot L, A)\)).
	Thus, \(q'\) causes \(q\) to remove itself from \(C'\) in the second step.
	Overall, we obtain \(C' = \VP(k \cdot S', A)\) within \(\BigO{\log \min \Set{k \cdot \ell, n}}\) rounds.
\end{proof}

Observe that this already yields efficient valid placement search procedures for shapes like parallelograms \(\Line(d_1, \ell_1) \oplus \Line(d_2, \ell_2)\) and trapezoids \(\Tri(d_1, \ell_1) \oplus \Line(d_2, \ell_2)\), and unions thereof.


\subsection{Shifted Shapes}
\label{subsec:segment_shifting}

For Minkowski sums of shapes with lines, we turn individual invalid placements into segments of invalid placements.
Now, we introduce a procedure that moves information with this structure along the segments' axis efficiently, as long as the segments have a sufficient length.

\begin{definition}
	\label{def:segmented_set}
	Let \(A\) be an amoebot structure, \(C \subseteq A\) a subset of amoebots, \(W \in \Set{X, Y, Z}\) a grid axis and \(k \in \Nats\).
	We call \(C\) a \emph{\(k\)-segmented set (on \(W\))} if on every maximal \(W\)-segment \(M\) of \(A\), where the maximal segments of \(C \cap M\) are \(C_1, \ldots, C_m\), the interior segments \(C_2, \ldots, C_{m-1}\) have length \(\geq k\).
\end{definition}
If a subset \(C\) of amoebots is \(k\)-segmented on the axis \(W\), we can move the set along this axis very efficiently by moving only the start and end points of the segments by \(k\) positions with the \PASC{} algorithm.
The size of the segments ensures that there is sufficient space between the \PASC{} start points to avoid interference.

\begin{lemma}
	\label{lem:segment_shifting}
	Let \(C \subseteq A\) be a \(k\)-segmented set on the axis \(W\) parallel to the direction \(d \in \Directions\) and let \(k\) be stored on a binary counter in \(A\).
	Given \(C\) and \(d\), the amoebots can compute the shifted set \(M \cap ((C \cap M) + k \cdot \UVec{d})\) on every maximal \(W\)-segment \(M\) of \(A\) within \(\BigO{\log \min \Set{k, n}}\) rounds.
	Furthermore, the resulting set of amoebots is \(k\)-segmented on \(W\).
\end{lemma}
\begin{proof}
	It suffices to show the lemma for arbitrary, individual maximal segments of \(A\) since all required properties are local to these segments.
	We synchronize the procedure across all segments using a global circuit.
	Without loss of generality, consider a maximal \(X\)-segment \(M\) of an amoebot structure \(A\) and let \(d = \W\) be the shifting direction.
	Let \(C \subseteq M\) be a \(k\)-segmented subset with segments \(C_1, \ldots, C_m\), ordered from West to East.
	To start with, we assume that all segments of \(C\) have length at least \(k\) and that \(M\) is large enough to fit all of \(C + k \cdot \UVec{d}\).
	
	Let \(p_1, \ldots, p_m\) and \(q_1, \ldots, q_m\) be the segments' westernmost and easternmost amoebots, respectively.
	These amoebots can identify themselves by checking which of their neighbors are contained in \(C\).
	We will call \(p_1, \ldots, p_m\) the start points and \(q_1, \ldots, q_m\) the end points of the segments.
	The amoebots now run the \PASC{} algorithm on the regions between the segment start points in direction \(d\) while transmitting \(k\) on the global circuit.
	This allows the amoebots \(p_i' = p_i + k \cdot \UVec{d}\) to identify themselves.
	The start points do not block each other because the distance between each pair of start points is greater than \(k\) due to the length of the segments \(C_i\).
	We repeat this for the end points to identify the amoebots \(q_i' = q_i + k \cdot \UVec{d}\).
	Finally, we establish circuits along \(M\) that are disconnected only at the new start and end points and let each \(q_i'\) beep in direction \(d\) to identify all amoebots in \(C + k \cdot \UVec{d}\).
	This procedure takes \(\BigO{\log \min \Set{d, n}}\) rounds by Lemma~\ref{lem:pasc_cutoff}.
	
	Now, we modify the algorithm to deal with the cases where \(M\) is not long enough and \(C_1\) and \(C_m\) have length less than \(k\).
	To handle the latter, we simply process \(C_1\) and \(C_m\) individually and run the procedure for \(C_2, \ldots, C_{m-1}\) as before.
	For this, \(p_1\) and \(q_m\) can identify themselves by using circuits to check whether there is another segment to the West or the East, respectively.
	
	Next, let \(u\) be the start point of \(M\), \ie{}, the westernmost amoebot of the segment.
	If the distance between \(u\) and \(p_1\) is at least \(k\), \(M\) is large enough to not interfere with the shift.
	Otherwise, \(u\) can decide whether it should become the start or end point of a shifted segment by comparing its distances to \(p_1\), \(q_1\), \(p_2\) and \(q_2\).
	In every possible case, \(u\) can uniquely decide which role it has to assume by comparing these distances to \(k\).
	For example, if the distance to \(p_1\) is less than \(k\) but the distance to \(q_1\) is greater than \(k\), then \(u\) becomes \(p_1'\).
	Because we only run the \PASC{} algorithm a constant number of times and use simple \(\BigO{1}\) circuit operations, the procedure takes \(\BigO{\log \min \Set{d, n}}\) rounds.
	Because every shifted segment maintains its length unless it runs into the start point of \(M\), in which case it becomes the first segment of the resulting set, we obtain a \(k\)-segmented set again.
\end{proof}

To leverage the efficiency of this procedure, we aim to construct shapes whose valid or invalid placements always form at least \(k\)-segmented sets at any scale \(k\).
The following property identifies such shapes.

\begin{definition}
	\label{def:wide_shapes}
	Let \(S\) be a shape and \(W \in \Set{X, Y, Z}\) a grid axis.
	The \emph{minimal axis width of \(S\) on \(W\)} (or \emph{\(W\)-width}) is the infimum of the lengths of all maximal components of the non-empty intersections of \(S\) with lines parallel to \(W\).
	We call \(S\) \emph{(\(W\)-)wide} or \emph{wide on \(W\)} if its \(W\)-width is at least \(1\).
\end{definition}

For example, the minimal axis width of \(\Tri(d, \ell)\) is \(0\) for all axes due to its corners and the \(W\)-width of \(\Line(d, \ell)\) is \(\ell\) when \(d\) is parallel to \(W\) and \(0\) otherwise.
If the \(W\)-width of \(S\) is \(w\), then the \(W\)-width of \(k \cdot S\) is \(k \cdot w\) for all scales \(k\).
Further, if \(S\) is \(W\)-wide, every node contained in \(S\) must have an incident edge parallel to \(W\) that is also contained in \(S\) and similarly, every face in \(S\) must have an adjacent face on \(W\) and every edge in \(S\) that is not parallel to \(W\) must have an incident face.

\begin{lemma}
	\label{lem:wide_shape_placements}
	Let \(S\) be a shape, \(W \in \Set{X, Y, Z}\) an axis and \(w \in \Nats\).
	If the minimal \(W\)-width of \(S\) is at least \(w\), then for all scales \(k\) and amoebot structures \(A\), \(A \setminus \VP(k \cdot S, A)\) is \(k \cdot w\)-segmented on \(W\).
	Conversely, if the \(W\)-width of \(S\) is \(0\), then for all scales \(k \geq 4\) there are amoebot structures \(A\) such that \(A \setminus \VP(k \cdot S, A)\) is at most \(1\)-segmented on \(W\) unless \(S\) is a single node.
\end{lemma}
%
\begin{proof}
	Let \(S\) be a shape with minimal \(W\)-width \(w\), \(k \in \Nats\) and \(A\) some amoebot structure.
	Consider a maximal \(W\)-segment \(M\) of \(A\) and let \(C = M \setminus \VP(k \cdot S, A)\).
	If \(C = \emptyset\), we are finished.
	Otherwise, let \(p \in C\) be arbitrary.
	Then there exists a node \(q \in V(k \cdot S + p) \setminus A\), \ie{}, node \(q\) is not occupied by an amoebot but it is occupied by \(k \cdot S\) placed at \(p\).
	Because of the \(W\)-width of \(S\), every component of every intersection of \(k \cdot S\) with a grid line parallel to \(W\) contains at least \(k \cdot w\) edges.
	Therefore, \(q\) lies on a \(W\)-segment of length at least \(k \cdot w\) that is occupied by \(k \cdot S + p\).
	Let \(Q\) be the set of placements of \(k \cdot S\) for which one node on this segment occupies \(q\).
	We then have \(p \in Q \cap M\) and \(Q \cap \VP(k \cdot S, A) = \emptyset\).
	Since both \(Q\) and \(M\) are \(W\)-segments, their intersection is also a \(W\)-segment.
	In the case \(Q \subseteq M\), \(p\) lies on a segment of \(C\) that contains \(Q\) and therefore has length at least \(k \cdot w\).
	In any other case, \(Q\) contains an endpoint of \(M\), which means that \(p\) lies on the first or the last segment of \(C\).
	Since this holds for every maximal \(W\)-segment of \(A\), \(A \setminus \VP(k \cdot S, A)\) is \(k \cdot w\)-segmented on \(W\).
	
	Now, let \(S\) be a non-trivial shape with a minimal \(W\)-width of \(0\) and let \(k \geq 4\) be arbitrary.
	We construct an amoebot structure \(A\) such that \(A \setminus \VP(k \cdot S, A)\) is at most \(1\)-segmented.
	First, we find a maximal \(W\)-segment \(L \subseteq V(k \cdot S)\) of \(V(k \cdot S)\) that contains at most two nodes and is not on the same \(W\)-line as the origin.
	If \(S\) has a node without incident edges on \(W\) that is not on the origin's line, we choose \(L\) as the scaled version of this node, as it will always be just a single node without neighbors on either side for scales \(k \geq 2\).
	Otherwise, if \(S\) has an edge that is not parallel to \(W\) and has no incident faces, we choose one of its middle nodes, which also has no neighbor on either side due to \(k \geq 4\).
	If this is also not the case, \(S\) must have a face \(f\) without an adjacent face on \(W\).
	Let \(u\) be the corner of \(f\) that is opposite of the face's edge on \(W\) and consider the two nodes adjacent to \(k \cdot u\) on the edges of \(k \cdot f\).
	Because \(f\) has no neighboring face on \(W\) and \(k \geq 3\), the segment spanning these two nodes is maximal in \(k \cdot S\).
	It also does not share the same \(W\)-line with the origin because it is offset from the scaled nodes of \(S\).
	
	We construct \(A\) by first placing a copy of \(V(k \cdot S)\).
	The origin is the only valid placement of \(k \cdot S\) in this structure.
	Next, we place another copy with the origin at position \(\ell \cdot \UVec{d}\), where \(\ell = \Abs{L} + 1\) and \(d\) is a direction parallel to \(W\), and add amoebots on the nodes \(V(\Line(d, \ell))\) to ensure connectivity.
	We thereby get another valid placement of \(k \cdot S\) at position \(\ell \cdot \UVec{d}\).
	Consider the node set \(L\), which was placed with the first copy of the shape.
	The node \(v\) one position in direction \(d\) of \(L\) remains unoccupied because \(L\) is not on the same \(W\)-line as the origin, is bounded by unoccupied nodes in \(V(k \cdot S)\) and its second copy is also placed such that one bounding node lies on \(v\).
	Thus, the amoebots \(i \cdot \UVec{d}\) for \(1 \leq i \leq \ell - 1\) are not valid placements of \(k \cdot S\) since they would require a copy of \(L\) that contains \(v\) to be occupied.
	We therefore have a maximal segment of \(A \setminus \VP(k \cdot S, A)\) that has length \(\Abs{L} - 1 \leq 1\).
	By repeating this construction two more times with sufficient distance in direction \(d\), adding lines on \(W\) to maintain connectivity, we obtain three such segments of invalid placements, one of which cannot be an outer segment.
	Since all of these segments lie on the same maximal segment of the resulting amoebot structure \(A\) (the one containing the origin), the set \(A \setminus \VP(k \cdot S, A)\) is at most \(1\)-segmented on \(W\).
\end{proof}

Lemma~\ref{lem:wide_shape_placements} allows us to apply the segment shift procedure (Lemma~\ref{lem:segment_shifting}) and move the invalid placements of \(k \cdot S\) along the axis \(W\) efficiently, as long as \(S\) is wide on \(W\).
This will become useful in conjunction with the fact that the Minkowski sum operation with a line \(\Line(d, \ell)\) always produces a shape of width at least \(\ell\) on the axis parallel to \(d\).


\begin{lemma}
	\label{lem:shifted_shapes}
	Let \(S\) be a \(W\)-wide shape and \(A\) an amoebot structure that stores a scale \(k\) in a binary counter and knows \(\VP(k \cdot S, A)\).
    Given a direction \(d\) on axis \(W\) and \(\ell \in \Nats\), the amoebots can compute \(\VP(k \cdot ((S + \ell \cdot \UVec{d}) \cup \Line(d, \ell)), A)\) within \(\BigO{\ell \cdot \log \min \Set{k \cdot \ell, n}}\) rounds.
\end{lemma}
%
\begin{proof}
	Let \(S\) be \(W\)-wide and consider an amoebot structure \(A\) that stores \(k\) on a binary counter and knows \(C = \VP(k \cdot S, A)\).
	Let \(d\) be parallel to \(W\) and \(\ell \in \Nats\) both be known by the amoebots and let \(S' = (S + \ell \cdot \UVec{d}) \cup \Line(d, \ell)\).
	
	By construction, \(k \cdot S'\) contains \(\Line(d, k \cdot \ell)\), so all invalid placements of the line are also invalid placements of \(k \cdot S'\).
	We assume that \(k\) is stored on a segment of \(A\) with maximum length, \eg{}, by using all maximal segments of \(A\) as counters simultaneously and deactivating counters whose space is exceeded by some operation.
	This way, when the amoebots compute \(k \cdot \ell\), at least one counter has enough space to store the result unless \(\VP(\Line(d, k \cdot \ell), A) = \emptyset\), in which case there are no valid placements of \(k \cdot S'\) and the amoebots can terminate.
	By Lemma~\ref{lem:line_detection}, the amoebots can now determine the valid placements of \(\Line(d, k \cdot \ell)\) within \(\BigO{\log \min \Set{k \cdot \ell, n}}\) rounds.
	
	Next, by Lemma~\ref{lem:wide_shape_placements}, the set \(\Tilde{C} = A \setminus C\) of invalid placements of \(k \cdot S\) is \(k\)-segmented on \(W\) in \(A\).
	Thus, the segment shift procedure (Lemma~\ref{lem:segment_shifting}) can be used to shift every amoebot \(q \in \tilde{C}\) by \(k\) positions in the opposite direction of \(d\) within its maximal \(W\)-segment of \(A\).
	Repeating this procedure \(\ell\) times, we shift \(q\) by \(k \cdot \ell\) positions and obtain the set \(Q \subset A\), where on every maximal \(W\)-segment \(M\) of \(A\), \(Q \cap M = M \cap ((M \cap \tilde{C}) - k \cdot \ell \cdot \UVec{d})\).
	For any \(q \in \tilde{C}\), we have \(q - k \cdot \ell \cdot \UVec{d} \notin \VP(k \cdot S', A)\) since \(V(k \cdot S) \subseteq V(k \cdot S' - k \cdot \ell \cdot \UVec{d})\) by the definition of \(S'\).
	Therefore, every amoebot in \(Q\) is an invalid placement of \(k \cdot S'\).
	Combining this with the line placement check, we get \(\VP(k \cdot S', A) \subseteq \VP(\Line(d, k \cdot \ell), A) \cap (A \setminus Q)\).
	
	Now, let \(q \in A \setminus \VP(k \cdot S', A)\) be arbitrary.
	If \(q \notin \VP(\Line(d, k \cdot \ell), A)\), \(q\) will be recognized by the line placement check.
	Otherwise, there must be a position \(x \in V(k \cdot S' + q)\) with \(x \notin A\).
	Since \(x \notin V(\Line(d, k \cdot \ell) + q)\), we have \(x \in V(k \cdot S + k \cdot \ell \cdot \UVec{d} + q)\).
	Therefore, \(q' = q + k \cdot \ell \cdot \UVec{d} \notin \VP(k \cdot S) = C\) and since \(q'\) lies on the same \(W\)-segment of \(A\) as \(q\), we have \(q \in Q\).
	This implies \(\VP(k \cdot S', A) = \VP(\Line(d, k \cdot \ell), A) \cap (A \setminus Q)\), which the amoebots have computed in \(\BigO{\ell \cdot \log \min \Set{k, n} + \log \min \Set{k \cdot \ell, n}} = \BigO{\ell \cdot \log \min \Set{k \cdot \ell, n}}\) rounds.
\end{proof}


\section{Shape Classification}
\label{sec:shapes}

As we have shown in Section~\ref{sec:lower_bound}, the transfer of valid placement information is not always possible in polylogarithmic time.
In this section, we combine the primitives described in the previous section to develop the class of \emph{snowflake} shapes, which always allow this placement information to be transmitted efficiently.
Additionally, we characterize the subset of \emph{star convex} shapes, for which the binary scale factor search is applicable.


\subsection{Snowflake Shapes}
\label{subsec:snowflakes}

Combining the primitives discussed in the previous section, we obtain the following class of shapes.
Our recursive definition identifies shapes with trees such that every node in the tree represents a shape and every edge represents a composition or transformation of shapes.

\begin{definition}
	\label{def:snowflakes}
	A \emph{snowflake tree} is a finite, non-empty tree \(T = (V_T, E_T)\) with three node labeling functions, \(\NodeType: V_T \to \Set{\TLine, \TTri, \TUnion, \TSum, \TShift}\), \(\NodeDir: V_T \to \Directions\) and \(\NodeLen: V_T \to \Nats_0\), that satisfies the following constraints.
	Every node \(v \in V_T\) represents a shape \(S_v\) such that:
	\begin{itemize}
		\item If \(\NodeType(v) = \TLine\), then \(v\) is a leaf node and \(S_v = \Line(\NodeDir(v), \NodeLen(v))\) \emph{(line node)}.
		\item If \(\NodeType(v) = \TTri\), then \(v\) is a leaf node and \(S_v = \Tri(\NodeDir(v), \NodeLen(v))\), where \(\NodeLen(v) > 0\) \emph{(triangle node)}.
		\item If \(\NodeType(v) = \TUnion\), then \(S_v = \bigcup_{i=1}^m S_{u_i}\), where \(u_1, \ldots, u_m\) are the children of \(v\) and \(m \geq 2\) \emph{(union node)}.
		\item If \(\NodeType(v) = \TSum\), then \(S_v = S_u \oplus \Line(\NodeDir(v), \NodeLen(v))\), where \(u\) is the unique child of \(v\) and \(\NodeLen(v) > 0\) \emph{(sum node)}.
		\item If \(\NodeType(v) = \TShift\), then \(S_v = (S_u + \NodeLen(v) \cdot \UVec{\NodeDir(v)}) \cup \Line(\NodeDir(v), \NodeLen(v))\), where \(u\) is the unique child of \(v\), \(S_u\) has a minimal axis width \(> 0\) on the axis of \(\NodeDir(v)\) and \(\NodeLen(v) > 0\) \emph{(shift node)}.
	\end{itemize}
	Let \(r \in V_T\) be the root of \(T\), then we say that \(S_r\) is the \emph{snowflake shape} represented by \(T\).
\end{definition}
Note that this definition constrains the placement of a snowflake's origin.
Because algorithms for the shape containment problem can place the origin of the target shape freely, we may extend the class of snowflakes to its closure under equivalence of shapes.
The algorithms we describe in this paper place the origin in accordance with the definition.




\subsection{Star Convex Shapes}
\label{subsec:star_convex_shapes}

The subset of \emph{star convex} shapes is of particular interest (see Fig.~\ref{fig:snowflakes} for example shapes):

\begin{definition}
    \label{def:star_convex_shapes}
    A shape \(S\) is \emph{star convex} if it is hole-free and contains a \emph{center node} \(c \in \Veqt\) such that for every \(v \in V(S)\), all shortest paths from \(c\) to \(v\) in \(\Geqt\) are contained in \(S\).
\end{definition}
For example, all convex shapes are star convex since all of their nodes are centers.
To show the properties of star convex shapes, we will use the following equivalent characterization:

\begin{lemma}
	\label{lem:star_convex_helper}
	A shape \(S\) is star convex with its origin as a center node if and only if \(S\) is the union of parallelograms of the form \(\Line(d, \ell) \oplus \Line(d', \ell')\) and convex shapes of the form \(\Tri(d, 1) \oplus \Line(d, \ell) \oplus \Line(d', \ell')\), where \(d'\) is obtained from \(d\) by a \(60^\circ\) clockwise rotation.
	The number of these shapes is in \(\BigO{\Abs{V(S)}}\).
\end{lemma}

\begin{figure}
    \centering
    \includegraphics[width=\linewidth]{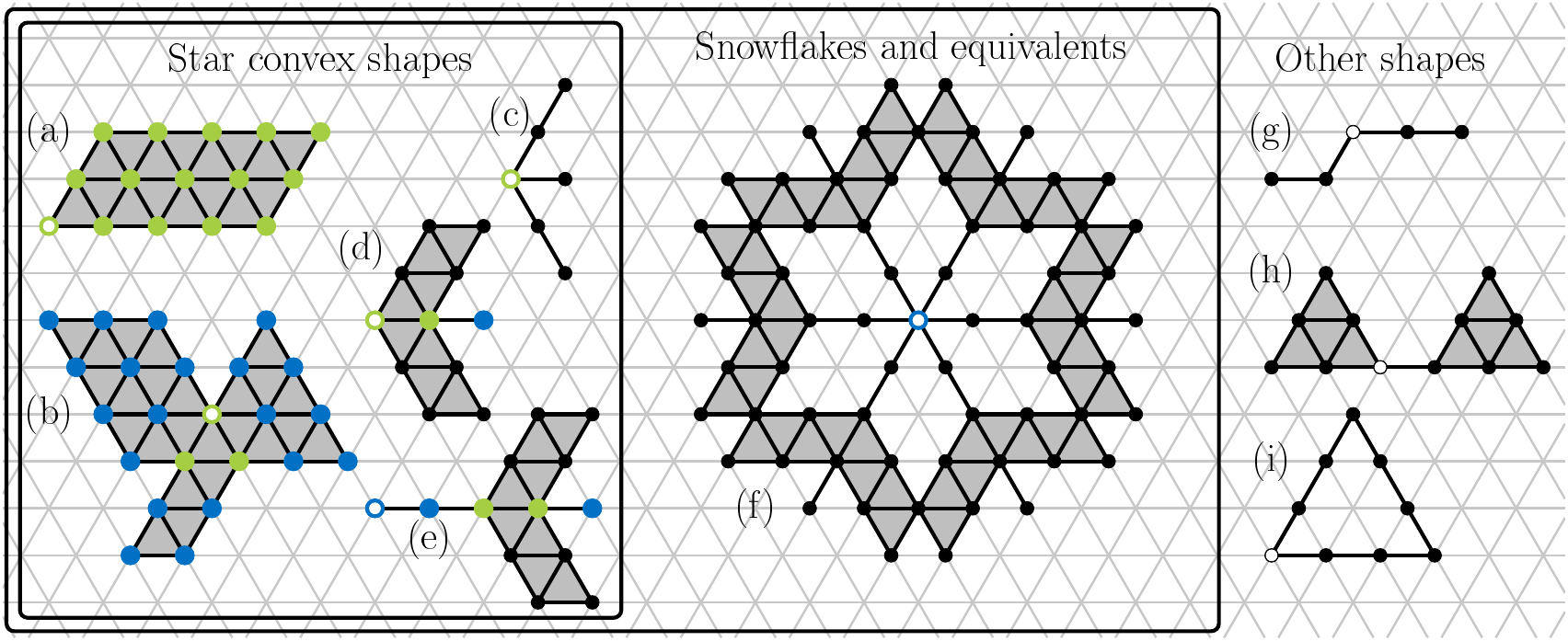}
    \caption{Examples of snowflakes, star convex shapes and other, non-snowflake shapes.
    Green nodes indicate star convex shape centers, blue nodes indicate possible snowflake origins and chosen origins are highlighted with a white center.
    All center nodes are also snowflake origins.
    Shape (a) is convex and shape (b) demonstrates that not all snowflake origins must be center nodes.
    Shape (c) is a union of lines, (d) is the Minkowski sum of (c) and \(\Line(\E, 1)\), (e) is the union of \(\Line(\E, 2)\) and \(\mathrm{(d)} + 2 \cdot \UVec{\E}\), and shape (f) consists of six rotated copies of (e).
    (g) is the example shape for the lower bound from Section~\ref{sec:lower_bound}.}
    \label{fig:snowflakes}
\end{figure}

Note that in the above lemma, \(d\), \(\ell\) and \(\ell'\) may not be the same for all constituent shapes.
The first kind of shape is always a line (if \(\ell = 0\) or \(\ell' = 0\)) or a parallelogram while the second kind of shape is a pentagon, a trapezoid (if \(\ell = 0\) or \(\ell' = 0\)) or a triangle (if \(\ell = \ell' = 0\)).
All of these shapes are convex.

\begin{proof}
	First, observe that every constituent shape \(\Line(d, \ell) \oplus \Line(d', \ell')\) or \(\Tri(d, 1) \oplus \Line(d, \ell) \oplus \Line(d', \ell')\) is convex and therefore, all of its nodes are center nodes, in particular its origin.
	When taking the union of such shapes, there is still at least one shortest path from the union's origin to every node because such a path is already contained in the convex shape that contributed the node.
	The union cannot have holes because for any hole, there would be a straight line from the origin to the boundary of the hole that does not lie completely inside the shape.
	This contradicts the fact that the convex shape contributing that part of the hole's boundary must already contain this line.
	Thus, the union is star convex and its origin is a center node.
	
	Now, let \(S\) be star convex such that the origin \(c\) is a center node.
	Consider any grid node \(v \in S\), then \(S\) contains all shortest paths between \(c\) and \(v\).
	Because shortest paths in the triangular grid always use only one or two directions at a \(60^\circ\) angle to each other, the set of shortest paths and the enclosed faces forms a parallelogram (see Fig.~\ref{subfig:shortest_paths_c}).
	By taking the union of all parallelograms spanned by \(c\) and the grid nodes in \(S\), we obtain at most \(\Abs{V(S)}\) parallelograms that already cover all grid nodes contained in \(S\) without adding extra elements.
	Any parallelogram of this kind can be represented as the Minkowski sum of two lines in the directions taken by the shortest paths.
	
	\begin{figure}
		\centering
		\begin{minipage}{0.34\textwidth}
			\centering
			\includegraphics[width=\textwidth]{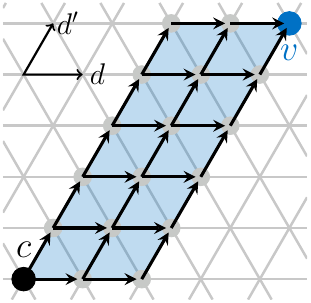}
			\subcaption{Shortest paths in the grid.}
			\label{subfig:shortest_paths_c}
		\end{minipage}
		\begin{minipage}{0.55\textwidth}
			\centering
			\includegraphics[width=\textwidth]{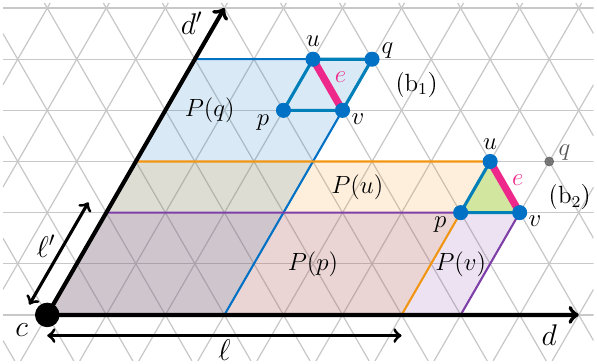}
			\subcaption{The convex shapes spanning a star convex shape.}
			\label{subfig:star_convex_union_proof}
		\end{minipage}
		\caption{Illustration of the proof of Lemma~\ref{lem:star_convex_helper}.
			(\subref{subfig:shortest_paths_c}) shows the parallelogram spanned by all shortest paths between nodes \(c\) and \(v\), using only the two directions \(d\) and \(d'\).
			(\subref{subfig:star_convex_union_proof}) shows the construction of the convex shapes containing all elements of a star convex shape.
			In (\(\mathrm{b_1}\)), the node \(q\) is contained in the shape, so the parallelogram \(P(q)\) already contains the edge \(e\) and its incident faces.
			In (\(\mathrm{b_2}\)), the node \(q\) is not contained, but the edge \(e\) and its incident face with \(p\) (in green) are contained in the Minkowski sum of \(P(p)\) and \(\Tri(d, 1)\), which does not add any elements outside of the shape due to \(P(u)\) and \(P(v)\).}
		\label{fig:star_convex_helper}
	\end{figure}
	
	The only remaining elements are edges that are not part of any shortest paths from \(c\) and their incident faces.
	Consider such an edge \(e \subset S\) with end points \(u, v\).
	Let \(p\) be the third corner of the face incident to \(e\) that is closer to \(c\) and let \(q\) be the third corner of the other incident face.
	If \(q \in S\), then \(e\) and the two faces are already contained in \(S\) because they lie in the parallelogram spanned by \(q\) (see Fig.~\ref{subfig:star_convex_union_proof}, \(\mathrm{b_1}\)).
	\(S\) must contain the two parallelograms spanned by \(u\) and \(v\) with \(c\).
	Since \(u\) and \(v\) have the same distance to \(c\) (otherwise \(e\) would be part of a shortest path), \(p\) must have a smaller distance and is therefore contained in both parallelograms.
	Thus, \(S\) contains the edges between \(p\) and \(u\) and between \(p\) and \(v\), which means that the face spanned by \(p\), \(u\) and \(v\) must be contained in \(S\) because \(S\) has no holes.
	Let \(P(p) = \Line(d, \ell) \oplus \Line(d', \ell')\) be the parallelogram spanned by \(p\), then \(P(u) = \Line(d, \ell) \oplus \Line(d', \ell' + 1)\) and \(P(v) = \Line(d, \ell + 1) \oplus \Line(d', \ell')\) (or vice versa).
	Now, \(\Tri(d, 1) \oplus P(p)\) is exactly the union of the face with the two parallelograms spanned by \(u\) and \(v\), \ie{}, it contains the face incident to \(e\) and does not add elements outside of \(S\) (see Fig.~\ref{subfig:star_convex_union_proof}, \(\mathrm{b_2}\)).
	Since the number of edges in \(S\) is bounded by \(6 \cdot \Abs{V(S)}\), the lemma follows.
\end{proof}
With this, we can relate star convex shapes to snowflake shapes as follows.

\begin{lemma}
    \label{lem:star_convex_snowflakes}
    Every star convex shape \(S\) is equivalent to a snowflake.
    If its origin is a center node, \(S\) itself is a snowflake.
\end{lemma}
\begin{proof}
	This follows directly from Lemma~\ref{lem:star_convex_helper} since all constituent shapes are Minkowski sums of lines and triangles.
	If the shape's origin is not a center, it can be translated so that it is a center, which results in an equivalent shape.
\end{proof}

A very useful property of star convex shapes is that they are self-contained.
We can even show that \emph{only} star convex shapes are self-contained.
As the authors of~\cite{hansen2020starshaped} point out in their extensive survey on \emph{starshaped sets} (see p.~1007), the results in~\cite{mcmullen1978sets} even show that these two properties are equivalent in much more general settings, when omitting rotations.

\begin{theorem}
	\label{theo:star_convex_self_contained}
	A shape is self-contained if and only if it is star convex.
\end{theorem}

In our context, this equivalence implies that the efficient binary search can only be applied directly to star convex shapes.
For any non-star convex shape \(S\), there exist an amoebot structure \(A\) and scale factors \(k < k'\) such that \(\VP(k \cdot S^{(r)}, A) = \emptyset\) for all \(r\) but \(\VP(k' \cdot S^{(r)}, A) \neq \emptyset\) for some \(r\); consider \eg{}, \(A = V(k' \cdot S)\) for sufficiently large \(k\) and \(k'\).


For the proof of Theorem~\ref{theo:star_convex_self_contained}, we first show several lemmas that provide the necessary tools.
To start with, we show that non-star convex shapes cannot become star convex by scaling, and for sufficiently large scale, every center candidate has a shortest path with a missing edge.
It is clear that star convex shapes stay star convex after scaling (by Lemma~\ref{lem:star_convex_helper}), but non-star convex shapes gain new potential center nodes, making it less obvious why there still cannot be a center.

\begin{lemma}
	\label{lem:shortest_path_missing_edge}
	Let \(S\) be an arbitrary non-star convex shape, then for every \(k \in \Nats\), \(k \cdot S\) is not star convex, and for \(k > 2\) and every node \(c \in V(k \cdot S)\), there exists a shortest path from \(c\) to a node \(v \in V(k \cdot S)\) in \(\Geqt\) with at least one edge not contained in \(k \cdot S\).
\end{lemma}
\begin{proof}
	Let \(S\) be an arbitrary non-star convex shape.
	If \(S\) contains a hole, all larger versions have a hole as well.
	We will construct the shortest paths for this case later.
	If \(S\) does not have any holes, then every node \(c \in V(S)\) has a shortest path \(\Path(c)\) to another node in \(S\) with at least one edge missing from \(S\).
	Consider a node \(c_k \in V(k \cdot S)\) for an arbitrary scale \(k \in \Nats\).
	If \(c_k = k \cdot c\) for \(c \in V(S)\), then the scaled path \(k \cdot \Path(c)\) has at least one edge that is not contained in \(k \cdot S\).
	Otherwise, \(c_k\) belongs to a scaled edge or face of \(S\).
	
	Consider the case \(c_k \in k \cdot e\) for some edge \(e \subseteq S\) with end points \(u, v \in V(S)\).
	Let \(d \in \Directions\) be the direction from \(u\) to \(v\).
	Since \(v\) is not a center of \(S\), there is a shortest path \(\Path = \Path(v)\) to some \(w \in V(S)\) with a missing edge.
	We choose this path such that its last edge is missing from \(S\) by pruning it after the first missing edge.
	Let \(\Path_k = k \cdot \Path\) be the scaled path from \(v_k = k \cdot v\) to \(w_k = k \cdot w\).
	Then, \(\Path_k\) is still a shortest path in \(\Geqt\) and its last \(k\) edges are not contained in \(k \cdot S\).
	We construct a shortest path \(\Path_k'\) from \(c_k\) to \(w_k\) that reaches \(\Path_k\) before it reaches \(w_k\).
	The construction depends on the directions occurring in \(\Path\) (see Fig.~\ref{fig:star_convex_path_cases_edges} for the following cases):
	If \(\Path\) only uses directions in \(\Set{d, d_L, d_R}\), where \(d_L\) and \(d_R\) are the counter-clockwise and clockwise neighbor directions of \(d\) (case (a) in the figure), we construct \(\Path_k'\) by adding the straight line from \(c_k\) to \(v_k\) at the beginning of \(\Path_k\); this line only uses edges in direction \(d\).
	If the first edge of \(\Path\) uses the opposite direction of \(d\) (case (b)), \(c_k\) already lies on \(\Path_k\) and we remove the straight line to \(v_k\) instead of adding it to obtain \(\Path_k'\).
	In all other cases, the first edge of \(\Path\) can only have directions other than \(d\) and its opposite.
	Due to symmetry, we only consider the cases where it uses \(d_R\) and its neighbor \(d_R' \neq d\).
	If the first edge of \(\Path\) has direction \(d_R\) (case (c)), the path must have an edge in direction \(d_R'\) (otherwise we are in the first case again).
	Then, the straight line from \(c_k\) in direction \(d_R\) will eventually reach the first scaled edge in direction \(d_R'\) on \(\Path_k\).
	If the first edge has direction \(d_R'\), there are two more cases:
	First, if \(\Path\) only uses directions in \(\Set{d_R', d_R}\), we again use a straight line in direction \(d_R\), which will meet the first scaled edge of \(\Path_k\) (case (d)).
	Otherwise, \(\Path\) must have an edge in the opposite direction of \(d\), which we will reach on \(\Path_k\) eventually with a straight line in direction \(d_R'\) (case (e)).
	
	\begin{figure}
		\centering
		\includegraphics[width=.9\linewidth]{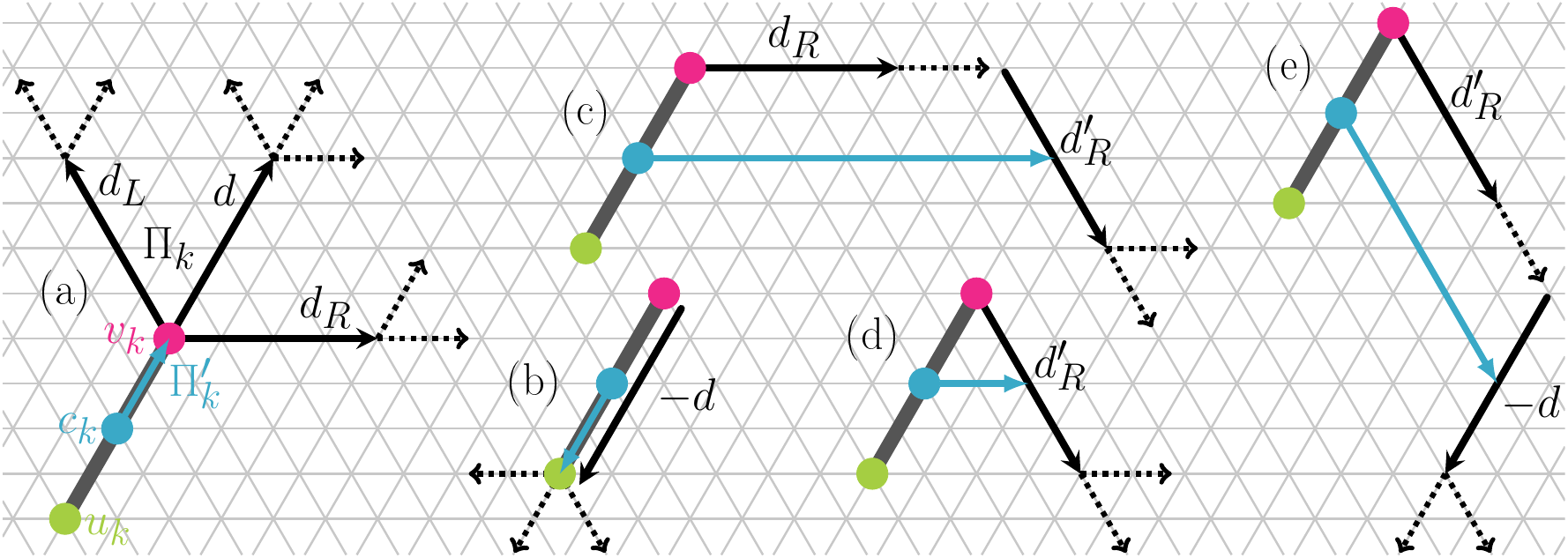}
		\caption{Overview of the possible cases for the shortest path \(\Path_k\) when the considered node \(c_k\) lies on a scaled edge, up to symmetry and rotation.
			The solid black arrows show the possible first edges of \(\Path_k\) and the dashed arrows show the directions in which the path can continue.
			In cases (c) and (e), the second solid arrow shows an edge that must exist somewhere on the path.
			The blue arrows show the first edges of \(\Path_k'\), constructed so that it reaches an edge of \(\Path_k\) before \(\Path_k\) reaches its end.
			The notation \(-d\) refers to the opposite direction of \(d\).}
		\label{fig:star_convex_path_cases_edges}
	\end{figure}
	
	\begin{figure}
		\centering
		\includegraphics[width=.9\linewidth]{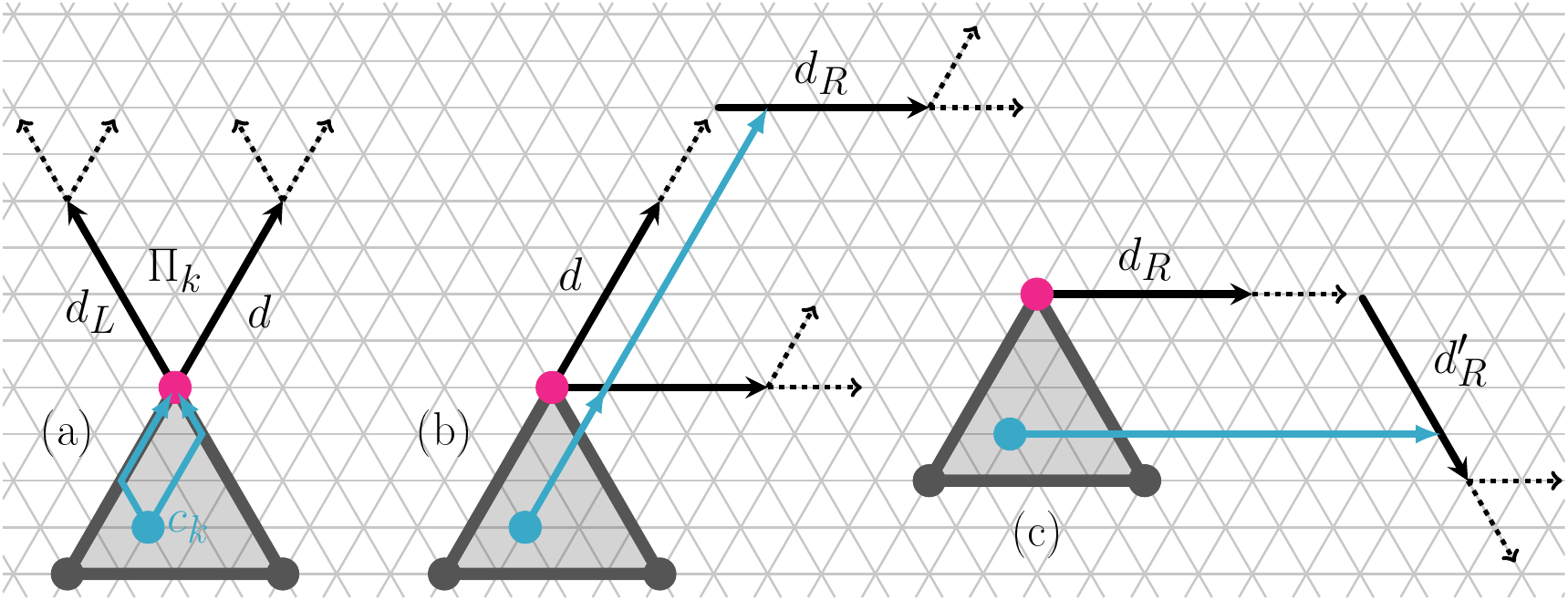}
		\caption{Overview of the cases where the considered node \(c_k\) lies in a scaled face, up to symmetry and rotation.
			The arrows have the same meaning as in Fig.~\ref{fig:star_convex_path_cases_edges}.
			Note that we can assume that the first edge of \(\Path\) is not one of the face's borders because otherwise, we can start the path at another corner of the face.}
		\label{fig:star_convex_path_cases_faces}
	\end{figure}
	
	For the case \(c_k \in k \cdot f\) for some face \(f \subseteq S\), we can use similar constructions from a path starting at a corner of the face (see Fig.~\ref{fig:star_convex_path_cases_faces}).
	This shows that \(k \cdot S\) is not star convex because no \(c_k \in V(k \cdot S)\) is a center node, for any \(k \in \Nats\), if \(S\) does not have a hole.
	In each case, we get a shortest path to another node that has at least one edge missing from \(k \cdot S\).
	
	\begin{figure}
		\centering
		\begin{minipage}{0.45\textwidth}
			\centering
			\includegraphics[width=\textwidth]{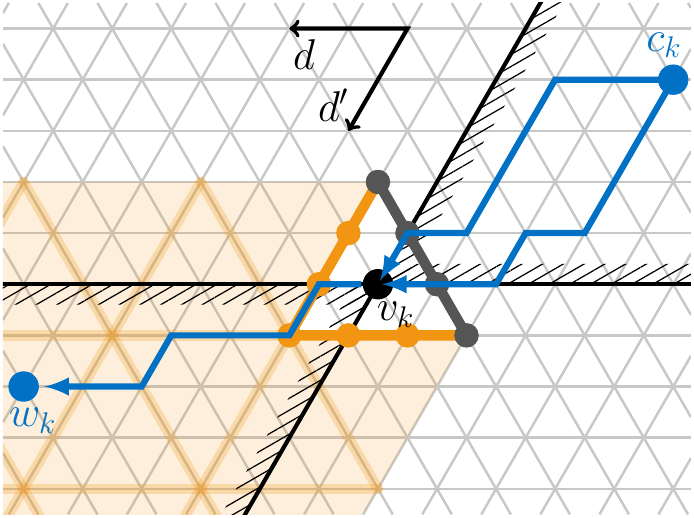}
			\subcaption{Shortest paths through a face.}
			\label{subfig:triangle_shortest_path_a}
		\end{minipage}
		\begin{minipage}{0.45\textwidth}
			\centering
			\includegraphics[width=\textwidth]{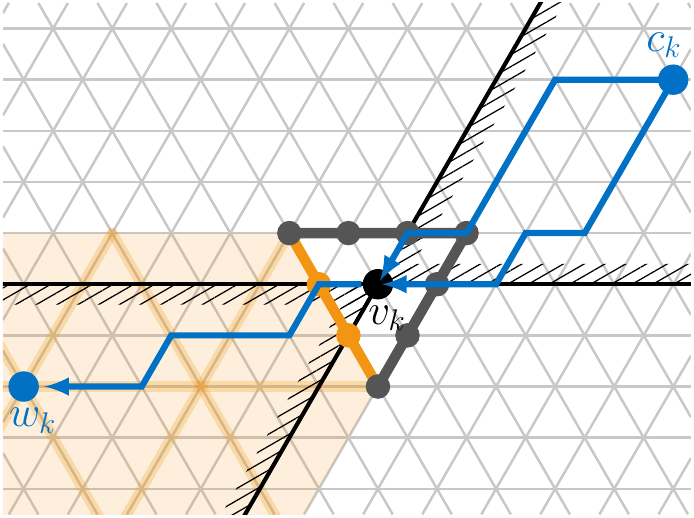}
			\subcaption{Alternative face orientation.}
			\label{subfig:triangle_shortest_path_b}
		\end{minipage}
		\caption{Construction of shortest paths through a hole of the shape at scale \(k = 3\).
			The triangle in the middle marks the scaled face \(k \cdot f\), where \(f\) is not contained in the shape.
			The shortest paths from some node \(c_k\) in the shape to node \(v_k\) in the hole use directions \(d\) and \(d'\) and are marked in blue.
			If they cannot be extended to another node \(w_k\) in the shape, then the shape cannot contain any element in the orange region, implying that the face \(f\) is not part of a hole because its region is not bounded by the shape.
			The black lines indicate the areas relative to \(v_k\) in which the nodes \(c_k\) and \(w_k\) can be located.
			By symmetry, the two depicted cases cover all possible positions of \(c_k\).}
		\label{fig:triangle_shortest_path}
	\end{figure}
	
	Finally, if \(S\) has a hole, consider any scale \(k \geq 3\).
	We construct the path with a missing edge as follows.
	Let \(f \subset \Reals^2\) with \(f \not\subseteq S\) be a face belonging to a hole of \(S\).
	Then, \(k \cdot f\) contains at least one node \(v_k \in \Geqt \setminus V(k \cdot S)\) (see Fig.~\ref{fig:triangle_shortest_path}).
	Thus, \(k \cdot S\) does not contain any of the edges incident to \(v_k\).
	Let \(c_k \in V(k \cdot S)\) be arbitrary and let the shortest paths from \(c_k\) to \(v_k\) use direction \(d\) and, optionally, \(d'\).
	Every shortest path from \(c_k\) to \(v_k\) uses one of the edges connecting \(v_k\) to \(v_k - \UVec{d}\) or \(v_k - \UVec{d'}\), both of which are not contained in \(k \cdot S\).
	If \(k \cdot S\) contains any node \(w_k\) that can be reached from \(v_k\) by a path using directions \(d\) and \(d'\) (or the other neighboring direction if \(d'\) was not used before), we can extend the path from \(c_k\) to \(v_k\) and obtain a shortest path from \(c_k\) to \(w_k\) that is missing at least one edge in \(k \cdot S\).
	If \(k \cdot S\) does not have any such node, the connected region of \(\Reals^2 \setminus k \cdot S\) that contains \(k \cdot f\) cannot be bounded by \(k \cdot S\), contradicting our assumption that \(f\) belongs to a hole of \(S\).
\end{proof}

Next, we show a \emph{fixed point} property that is particularly useful for scales \(k\) and \(k + 1\).

\begin{lemma}
	\label{lem:fixed_point_properties}
	Let \(S\) be a shape and \(k \in \Nats\) a scale such that there exists a \(t \in \Veqt\) with \(k \cdot S + t \subseteq (k+1) \cdot S\).
	Then, \(t\) must be in \(V(S)\) and we call \(t\) a \emph{fixed node} of \(S\).
	Furthermore, for every node \(v \in V(S)\), a shortest path from \(k \cdot v + t\) to \((k+1) \cdot v\) is also a shortest path from \(t\) to \(v\) and vice versa.
\end{lemma}
\begin{proof}
	Consider a shape \(S\), a scale \(k \in \Nats\) and a translation \(t \in \Veqt\) with \(k \cdot S + t \subseteq (k+1) \cdot S\) (see Fig.~\ref{fig:overlapping_node} for an illustration of the lemma).
	First, observe that the only node that is mapped to the same position by the two transformations of \(S\) is \(t\) itself:
	\[
	k \cdot x + t = (k+1) \cdot x
	\iff
	t = x
	\]
	Next, suppose \(t \notin V(S)\) and let \(v \in V(S)\) be a node of \(S\) with minimum grid distance \(\ell \in \Nats\) to \(t\).
	Let \(t_k = k \cdot t + t = (k+1) \cdot t\) be the transformed position of \(t\).
	Because our scaling operation for shapes is uniform, the grid distance between \(t_k\) and \(v_k = k \cdot v + t\) is \(k \cdot \ell\).
	However, the distance between \(t_k\) and any closest node of \((k+1) \cdot S\) (\eg{}, \(v_{k+1} = (k+1) \cdot v\)) is \((k+1) \cdot \ell > k \cdot \ell\).
	Therefore, \(v_k\) cannot be contained in \((k+1) \cdot S\), contradicting our assumption for \(t\).
	Thus, \(t \in V(S)\).
	
	\begin{figure}
		\centering
		\begin{minipage}[t]{0.35\textwidth}
			\centering
			\includegraphics[width=\textwidth]{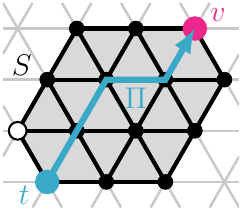}
			\subcaption{Example shape \(S\).}
			\label{subfig:overlapping_node_a}
		\end{minipage}
		\begin{minipage}[t]{0.5\textwidth}
			\centering
			\includegraphics[width=\textwidth]{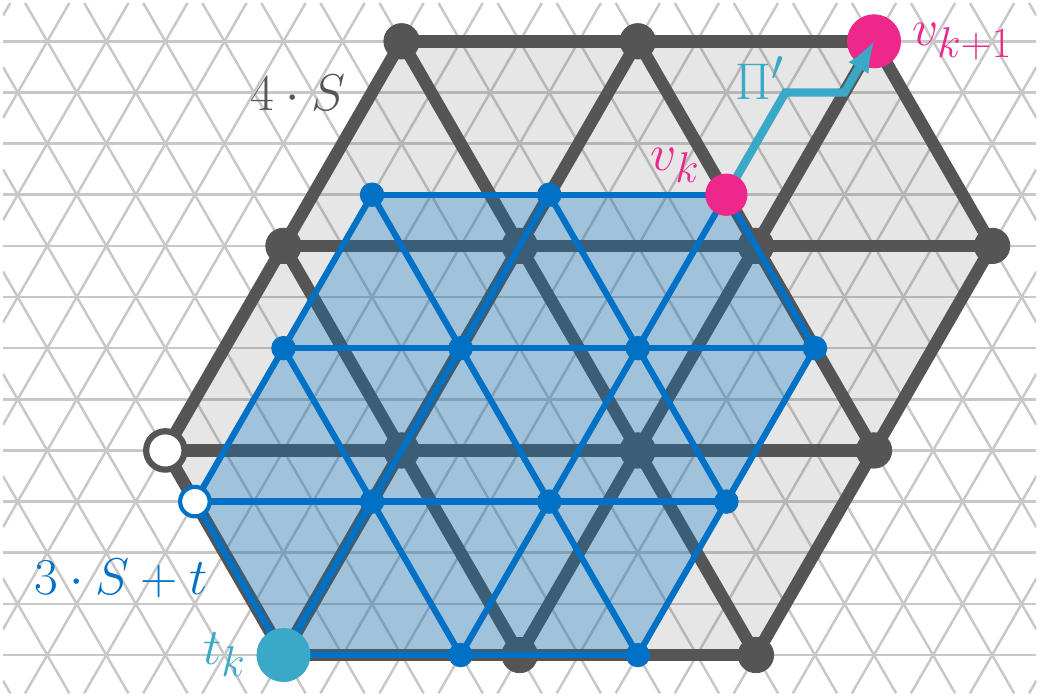}
			\subcaption{A placement with \(3 \cdot S + t \subseteq 4 \cdot S\).}
			\label{subfig:overlapping_node_b}
		\end{minipage}
		\caption{Illustration of fixed nodes and shortest paths between scaled nodes.
			(\subref{subfig:overlapping_node_a}) shows the shape \(S\) with two highlighted nodes \(t\) and \(v\) as well as a shortest path \(\Path\) between them.
			(\subref{subfig:overlapping_node_b}) shows the shapes \(3 \cdot S + t\) and \(4 \cdot S\), where \(3 \cdot S + t \subset 4 \cdot S\).
			The scaled nodes \(t_k = 3 \cdot t + t\), \(v_k = 3 \cdot v + t\) and \(v_{k+1} = 4 \cdot v\) are highlighted.
			The shortest path \(\Path'\) from \(v_k\) to \(v_{k+1}\) is a translated version of \(\Path\).}
		\label{fig:overlapping_node}
	\end{figure}
	
	
	For the shortest path property, consider some node \(v \in V(S)\) and let \(\Path\) be a shortest path from \(t\) to \(v\) in the grid.
	By translating this path by the vector \(k \cdot v\), we obtain the path \(\Path' = \Path + k \cdot v\), which starts at \(k \cdot v + t\), ends at \(k \cdot v + v = (k+1) \cdot v\) and is still a shortest path.
	The same construction works the other way around by subtracting \(k \cdot v\) instead.
\end{proof}

Finally, we eliminate the need for covering rotations by showing that for sufficiently large scales \(k\) and \(k+1\), \(k \cdot S^{(r)}\) does not fit into \((k+1) \cdot S\) for any \(r \in \Set{1,\ldots,5}\) unless \(S\) is rotationally symmetric.

\begin{definition}
	\label{def:symmetry}
	A shape \(S\) is called \emph{rotationally symmetric with respect to \(r \in \Set{1, 2, 3}\)} (or \emph{\(r\)-symmetric}) if there exists a translation \(t \in \Veqt\) such that \(S^{(r)} + t = S\).
\end{definition}
Note that \(r \in \Set{1, 2, 3}\) covers all possible rotational symmetries in the triangular grid, and \(1\)-symmetry is equivalent to \(2\)- and \(3\)-symmetry combined.
Furthermore, the translation \(t\) is unique.
Also note that \(1\)-symmetry is more commonly called \emph{\(6\)-fold} symmetry, \(2\)-symmetry is known as \(3\)-fold symmetry and \(3\)-symmetry is known as \(2\)-fold symmetry.

To prove the following lemma about rotations, we will use a simple set of conditions for a convex shape to fit into another, which we introduce first.
For that, observe that in the triangular grid, every convex shape has at most six sides and is uniquely defined by the lengths of these sides.
For example, for the single node, all side lengths are \(0\), for a regular hexagon, all side lengths are equal, and for a triangle of size \(\ell\), the side lengths alternate between \(0\) and \(\ell\).

\begin{lemma}
	\label{lem:convex_shape_inequalities}
	Let \(S_1, S_2\) be two convex shapes whose side lengths are \(a_1,\ldots,f_1 \in \Nats_0\) and \(a_2,\ldots,f_2 \in \Nats_0\), respectively, ordered in counter-clockwise direction around the shape, and \(a_i\) is the length of the bottom left side, parallel to the \(Z\)-axis.
	Then, \(S_1\) can be placed inside \(S_2\), \ie{}, there exists a \(t \in \Veqt\) such that \(S_1 + t \subseteq S_2\), if and only if the following inequalities hold:
	\begin{flalign}
		&&&a_1 + b_1 &&\leq &&a_2 + b_2 &&&&&&\label{eq:1}\\
		&&&b_1 + c_1 &&\leq &&b_2 + c_2 &&&&&&\label{eq:2}\\
		&&&c_1 + d_1 &&\leq &&c_2 + d_2 &&&&&&\label{eq:3}\\
		&&&a_1 + b_1 + c_1 &&\leq &&a_2 + b_2 + c_2 &&&&&&\label{eq:4}\\
		&&&b_1 + c_1 + d_1 &&\leq &&b_2 + c_2 + d_2 &&&&&&\label{eq:5}
	\end{flalign}
\end{lemma}
\begin{proof}
	We refer to the sides of the two shapes as \(\Side{A}_i, \Side{B}_i\), \etc{} for \(i \in \Set{1, 2}\) and assume \Wlog{} that the origin of shape \(i\) lies at the position where sides \(\Side{A}_i\) and \(\Side{B}_i\) meet (see Fig.~\ref{subfig:convex_shape_outline}).
	To show the lemma, we have to prove that the inequalities are both necessary and sufficient.
	
	\subparagraph{Necessary condition:}
	Inequalities \eqref{eq:1}--\eqref{eq:3} relate the distances between the parallel sides of \(S_1\) and \(S_2\).
	If one of them does not hold, the distance between two parallel sides of \(S_1\) is greater than the distance between the two corresponding sides of \(S_2\), making it impossible for \(S_1\) to fit into \(S_2\) (see Fig.~\ref{subfig:convex_shape_proof_a}).
	Therefore, these inequalities are necessary.
	Next, if \(S_1\) fits into \(S_2\), we can always translate it so that it is still contained and \(\Side{A}_1\) intersects \(\Side{A}_2\), \(\Side{B}_1\) intersects \(\Side{B}_2\), or both (in which case \(t = 0\)).
	
	Case 1: Both sides can intersect (see Fig.~\ref{subfig:convex_shape_proof_b}).
	Then, we must have \(a_1 \leq a_2\) and \(b_1 \leq b_2\) because otherwise, at least one of the sides would reach outside of \(S_2\).
	Combining this with inequalities \eqref{eq:2} and \eqref{eq:3} yields the last two inequalities.
	
	Case 2: Only \(\Side{B}_1\) and \(\Side{B}_2\) can intersect (see Fig.~\ref{subfig:convex_shape_proof_c}).
	Then, we have \(b_1 \leq b_2\) and \(a_1 > a_2\).
	This already yields \eqref{eq:5} due to \(c_1 + d_1 \leq c_2 + d_2\).
	Now, we shift \(S_1\) parallel to \(\Side{B}_1\) such that it is still contained in \(S_2\) and \(\Side{F}_1\) intersects \(\Side{F}_2\).
	This always works because \(\Side{A}_2\) is not in the way, otherwise we would be in case 1 again.
	Then, the distance between the sides \(\Side{A}_1\) and \(\Side{A}_2\) is \(a_1 - a_2\) and the distance between \(\Side{D}_1\) and \(\Side{A}_2\) is \(c_1 + b_1 + a_1 - a_2\).
	Since this distance must be at most \(b_2 + c_2\) (the distance between parallel sides of \(S_2\)), we get \(a_1 + b_1 + c_1 - a_2 \leq b_2 + c_2 \iff a_1 + b_1 + c_1 \leq a_2 + b_2 + c_2\).
	Case 3, where only \(\Side{A}_1\) and \(\Side{A}_2\) can intersect, is analogous to case 2.
	
	\begin{figure}
		\centering
		\begin{minipage}[t]{0.4\textwidth}
			\centering
			\includegraphics[width=\textwidth]{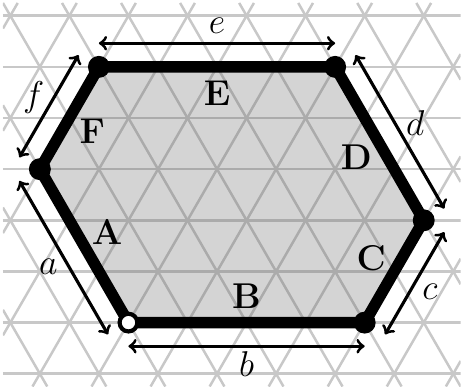}
			\subcaption{Convex shapes in the triangular grid.}
			\label{subfig:convex_shape_outline}
		\end{minipage}
		\hfill
		\begin{minipage}[t]{0.4\textwidth}
			\centering
			\includegraphics[width=\textwidth]{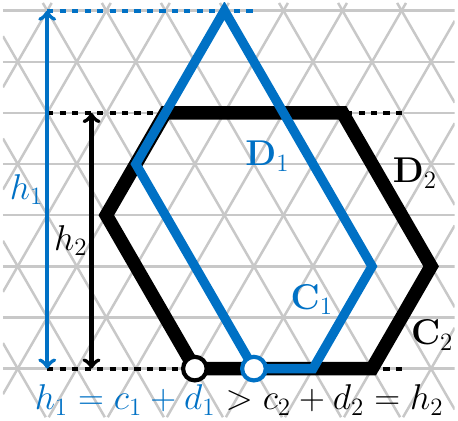}
			\subcaption{Violation of inequality~\eqref{eq:3}}
			\label{subfig:convex_shape_proof_a}
		\end{minipage}

        \bigskip
		\begin{minipage}[t]{0.35\textwidth}
			\centering
			\includegraphics[width=\textwidth]{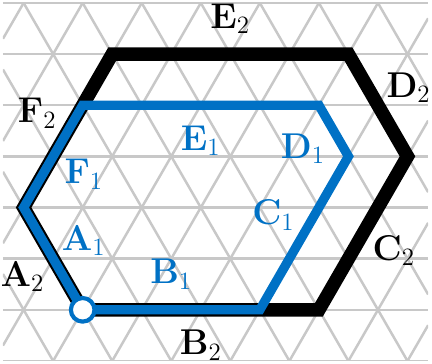}
			\subcaption{Simple case for \(S_1 \subseteq S_2\).}
			\label{subfig:convex_shape_proof_b}
		\end{minipage}
		\hfill
		\begin{minipage}[t]{0.4\textwidth}
			\centering
			\includegraphics[width=\textwidth]{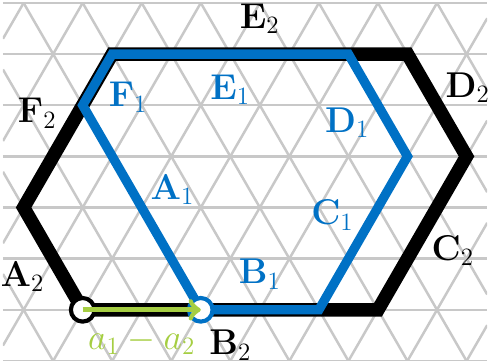}
			\subcaption{Harder case for \(S_1 + t \subseteq S_2\).}
			\label{subfig:convex_shape_proof_c}
		\end{minipage}
		\caption{Reference for the proof of Lemma~\ref{lem:convex_shape_inequalities}.
			(\subref{subfig:convex_shape_outline}) shows the general outline of a convex shape.
			We refer to the six sides as \(\Side{A}, \ldots, \Side{F}\) and their respective lengths as \(a, \ldots, f\).
			The side lengths can be \(0\), resulting in sharper corners.
			Observe that some constraints apply, \eg{}, \(a + f = c + d\).
			In the other figures, the outline of \(S_1\) is shown in blue and \(S_2\) is shown in black.
			(\subref{subfig:convex_shape_proof_a}) shows a situation where the third inequality is violated and \(S_1\) does not fit into \(S_2\).
			(\subref{subfig:convex_shape_proof_b}) shows a situation where the \(\Side{A}_i \cap \Side{B}_i\) corners of the two shapes (their origins) coincide and the side lengths \(a_i\) and \(b_i\) permit this placement.
			In (\subref{subfig:convex_shape_proof_c}), \(S_1\) has been moved to the right by \(a_1 - a_2\) steps because \(a_1 > a_2\) prevents the placement with aligning origins.
			Observe that \(S_1\) can always be moved left or right such that \(\Side{A}_1\) and \(\Side{A}_2\) or \(\Side{F}_1\) and \(\Side{F}_2\) intersect.}
		\label{fig:convex_shape_proof}
	\end{figure}
	
	\subparagraph{Sufficient condition:}
	Suppose all inequalities are satisfied.
	We place \(S_1\) at \(t = 0\) so that the \(\Side{A}_i \cap \Side{B}_i\) corners align (see Fig.~\ref{subfig:convex_shape_proof_b}).
	
	Case 1: \(a_1 \leq a_2\) and \(b_1 \leq b_2\).
	In this case, the corners \(\Side{A}_1 \cap \Side{B}_1\), \(\Side{A}_1 \cap \Side{F}_1\) and \(\Side{B}_1 \cap \Side{C}_1\) are in \(S_2\) because they lie directly on \(\Side{A}_2\) or \(\Side{B}_2\).
	Combining this with \eqref{eq:2} and \eqref{eq:3}, the \(\Side{C}_1 \cap \Side{D}_1\) and \(\Side{E}_1 \cap \Side{F}_1\) corners are also in \(S_2\), respectively (the latter is due to the constraint \(a + f = c + d\)).
	Because the sides \(\Side{D}_1\) and \(\Side{E}_1\) form a shortest path between the \(\Side{E}_1 \cap \Side{F}_1\) corner and the \(\Side{C}_1 \cap \Side{D}_1\) corner, and since \(S_2\) is convex, the \(\Side{D}_1 \cap \Side{E}_1\) corner must also be in \(S_2\).
	Since all corners of \(S_1\) are in \(S_2\), we have \(S_1 \subseteq S_2\).
	
	Case 2: \(b_1 \leq b_2\) but \(a_1 > a_2\).
	Then, we move \(S_1\) by \(a_1 - a_2\) positions to the right, parallel to \(\Side{B}_1\) (see Fig.~\ref{subfig:convex_shape_proof_c}).
	Now, the two sides \(\Side{A}_1\) and \(\Side{B}_1\) are in \(S_2\), unless \(b_2 - b_1 < a_1 - a_2\) or \(a_1 > a_2 + f_2\), which would violate \eqref{eq:1} or \eqref{eq:2}, respectively (the second one follows because \(a_2 + f_2 = c_2 + d_2\) and \(a_1 \leq c_1 + d_1\)).
	The \(\Side{C}_1 \cap \Side{D}_1\) corner's distance to the origin (the \(\Side{A}_2 \cap \Side{B}_2\) corner) is now \(a_1 - a_2 + b_1 + c_1\).
	The corner has not moved past the \(\Side{D}_2\) line since the distance between \(\Side{A}_2\) and \(\Side{D}_2\) is \(b_2 + c_2\) and \(a_1 + b_1 + c_1 \leq a_2 + b_2 + c_2\) implies \(a_1 - a_2 + b_1 + c_1 \leq b_2 + c_2\).
	Thus, \(\Side{C}_1\) is in \(S_2\).
	The \(\Side{E}_1 \cap \Side{F}_1\) corner is below the \(\Side{E}_2\) line due to \(a_1 + f_1 = c_1 + d_1 \leq c_2 + d_2 = a_2 + f_2\).
	It has to lie on \(\Side{F}_2\) because \(\Side{A}_1\) still intersects \(\Side{F}_2\) after the shift.
	Thus, the \(\Side{E}_1 \cap \Side{F}_1\) corner is in \(S_2\).
	Since five corner points of \(S_1\) are in \(S_2\) and \(\Side{D}_1\) and \(\Side{E}_1\) form a shortest path between two of those points, all points of \(S_1\) are in \(S_2\).
	Again, the third case (\(b_1 > b_2\) and \(a_1 \leq a_2\)) is analogous to the second case.
\end{proof}

\begin{lemma}
	\label{lem:rotation_symmetry}
	Let \(S\) be a shape and \(r \in \Set{1, 2, 3}\) be arbitrary.
	If \(S\) is not \(r\)-symmetric, then there is a scale \(k_0 \in \Nats\) such that for every \(k \geq k_0\), \(k \cdot S^{(r)}\) does not fit into \((k+1) \cdot S\).
\end{lemma}
\begin{proof}
	First, consider two convex shapes \(S, S'\) with side lengths \(a,\ldots,f\) and \(a',\ldots,f'\), respectively.
	By Lemma~\ref{lem:convex_shape_inequalities}, for every \(k \in \Nats\), \(k \cdot S'\) fits into \((k+1) \cdot S\) if and only if the inequalities \(k \cdot L_i \leq (k+1) \cdot R_i\) for \(1 \leq i \leq 5\) are satisfied, where \(L_1 = a' + b'\), \(R_1 = a + b\) \etc{}
	We show that for sufficiently large \(k\), we can remove the factors \(k\) and \((k+1)\) from the inequalities.
	For one direction, we already have \(L_i \leq R_i \implies k \cdot L_i \leq (k+1) \cdot R_i\).
	To get the other direction, choose \(k > R_i\), then we have
	\begin{align*}
		& k \cdot L_i \leq (k + 1) \cdot R_i\\
		\implies\qquad & L_i \leq \frac{k + 1}{k} R_i = R_i + \frac{R_i}{k} < R_i + 1,~\text{by choice of}~k\\
		\implies\qquad & L_i \leq R_i,~\text{since}~ L_i, R_i \in \Nats_0.
	\end{align*}
    Thus, for $k > R_1, \ldots, R_5$, we get $k \cdot L_i \leq (k+1) \cdot R_i \iff L_i \leq R_i$.
	
	Next, let \(S\) be an arbitrary shape and let \(H\) be its \emph{convex hull}, \ie{}, the smallest convex shape that contains \(S\).
	We will show that for every \(r \in \Set{1, 2, 3}\), \(H\) must be \(r\)-symmetric if \(k \cdot H^{(r)}\) fits into \((k+1) \cdot H\) for infinitely many scales \(k\).
	Let \(a,\ldots,f\) be the side lengths of \(H\) and observe that for every rotation \(r\), \(H^{(r)}\) is a convex shape with side lengths \(a', \ldots, f'\) that are a simple reordering of \(a, \ldots, f\).
	For \(r = 1\), we have \(a' = f\), \(b' = a\), \(c' = b\), \(d' = c\) \etc{}
	Suppose \(k \cdot H^{(r)}\) fits into \((k+1) \cdot H\) for infinitely many \(k\), then the resulting inequalities \(L_i \leq R_i\) for \(H^{(r)}\) and \(H\) must be satisfied, as shown above.
	The inequalities for \(r = 1\) are the following:
	\begin{align*}
		a' + b' = f + a = c + d                 \enspace&\leq\enspace a + b,\\
		b' + c' = a + b                         \enspace&\leq\enspace b + c,\\
		c' + d' = b + c                         \enspace&\leq\enspace c + d,\\
		a' + b' + c' = f + a + b = b + c + d    \enspace&\leq\enspace a + b + c,\\
		b' + c' + d' = a + b + c                \enspace&\leq\enspace b + c + d
	\end{align*}
	The first three inequalities imply \(a + b = b + c = c + d\), leading to \(a = c\) and \(b = d\).
	The last two inequalities imply \(a = d\), meaning that \(a = b = \ldots = f\) is the only way to satisfy all inequalities.
	Thus, for \(r = 1\), \(H\) must be a regular hexagon, which is \(1\)-symmetric.
	Fig.~\ref{fig:rotational_symmetry} illustrates this and the remaining two cases.
	
	\begin{figure}
		\centering
		\includegraphics[width=0.9\linewidth]{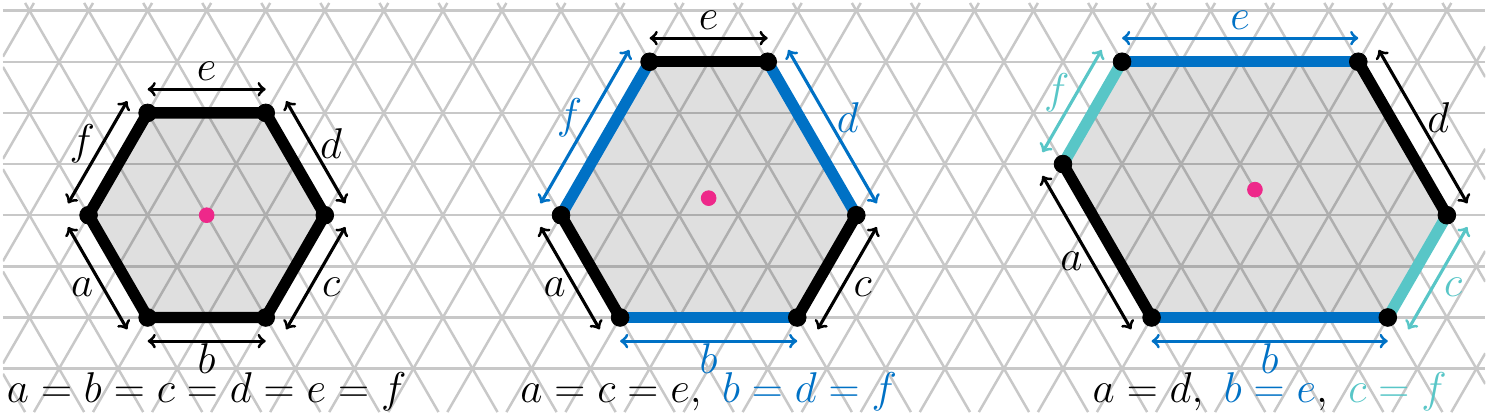}
		\caption{The three types of rotationally symmetric shapes in the triangular grid.
			The left shape is \(1\)-symmetric, the middle shape is \(2\)-symmetric and the right shape is \(3\)-symmetric.
			In each shape, the red dot indicates the center of rotation, which might not lie on a grid node.}
		\label{fig:rotational_symmetry}
	\end{figure}
	
	For \(r = 2\), the side lengths of \(H^{(r)}\) are \(a' = e\), \(b' = f\), \(c' = a\), \(d' = b\), \(e' = c\) and \(f' = d\), resulting in the following inequalities:
	\begin{align*}
		e + f = b + c           \enspace&\leq\enspace a + b,\\
		f + a = c + d           \enspace&\leq\enspace b + c,\\
		a + b                   \enspace&\leq\enspace c + d,\\
		e + f + a = a + b + c   \enspace&\leq\enspace a + b + c,\\
		f + a + b = b + c + d   \enspace&\leq\enspace b + c + d
	\end{align*}
	The first three inequalities imply \(a = c\) and \(b = d\), as before, but the last two inequalities are tautologies.
	However, observe that \(e = a + b - d = a\) and \(f = c + d - a = d\) due to the natural constraints on the side lengths.
	Thus, \(H\) has only two distinct side lengths that alternate along its outline, making it \(2\)-symmetric.
	
	Finally, for \(r = 3\), we get the side lengths \(a' = d\), \(b' = e\), \(c' = f\), \(d' = a\), \(e' = b\) and \(f' = c\).
	The resulting inequalities are:
	\begin{align*}
		d + e = a + b           \enspace&\leq\enspace a + b,\\
		e + f = b + c           \enspace&\leq\enspace b + c,\\
		f + a = c + d           \enspace&\leq\enspace c + d,\\
		d + e + f = b + c + d   \enspace&\leq\enspace a + b + c,\\
		e + f + a = a + b + c   \enspace&\leq\enspace b + c + d
	\end{align*}
	Since the first three inequalities add no information, we only get \(a + b + c = b + c + d\) from the last two inequalities, \ie{}, \(a = d\).
	From this, we can deduce \(b = e + d - a = e\) and \(c = a + f - d = f\), so the opposite sides of \(H\) must have the same lengths.
	Such a shape is \(3\)-symmetric.
	
	This already shows the lemma for convex shapes and the convex hull \(H\) of \(S\) in particular.
	Now, suppose that for some \(r \in \Set{1, 2, 3}\), \(k \cdot S^{(r)}\) fits into \((k+1) \cdot S\) for infinitely many \(k \in \Nats\).
	As a necessary condition, the same property holds for \(H\), so \(H\) must be \(r\)-symmetric.
	Let \(\pi: \Reals^2 \to \Reals^2\) be the composition of a rotation by \(r\) and translation by \(t \in \Veqt\) that maps \(H\) into itself, \ie{}, \(\pi(H) = H^{(r)} + t = H\) (see Fig.~\ref{subfig:rotation_mapping_a}).
	\(\pi\) is an in-place rotation that does not alter \(H\) but might change \(S\) if it is not \(r\)-symmetric.
	Let \(d\) be the diameter of \(H\) (the largest grid distance between any two nodes in \(H\)) and consider a scale \(k > d\) such that \(k \cdot S^{(r)} + t' \subseteq (k+1) \cdot S\) for some \(t' \in \Veqt\).
	Further, let \(\Map: \Reals^2 \to \Reals^2\) be the composition of a rotation, scaling and translation that transforms \(S\) into \(k \cdot S^{(r)} + t' \subseteq (k+1) \cdot S\).
	We get \(\Map(H) \subseteq (k+1) \cdot H\) and by Lemma~\ref{lem:fixed_point_properties}, for every \(v \in V(H)\), the distance between \((k+1) \cdot v\) and \(\Map(\pi^{-1}(v))\) equals the distance between two nodes in \(H\) and is thus bounded by \(d\).
	The lemma is applicable because \(\pi\) is a bijection on \(V(H)\) which makes \(\Map(\pi^{-1}(\cdot))\) a composition of a scaling by \(k\) and a translation.
	
	Now, let \(x \in V(H) \setminus S\) be a node in \(H\) that is not part of \(S\).
	If no such \(x\) exists, then either \(S = H\), in which case we are done, or \(S\) and \(H\) only differ in their edges and faces, which we will handle later.
	Since \(x \notin S\), \(S\) does not contain any edges or faces incident to \(x\).
	Thus, \((k+1) \cdot S\) does not contain any nodes with distance less than \(k+1\) to \((k+1) \cdot x\).
	Consider the node \(y = \pi^{-1}(x) \in V(H)\) that is rotated onto \(x\) (see Fig.~\ref{subfig:rotation_mapping_b}).
	As explained above, the distance between \(\Map(y)\) and \((k+1) \cdot S\) is at most \(d < k+1\), so we have \(\Map(y) \notin (k+1) \cdot S\).
	By our assumption that \(\Map(S) \subseteq (k+1) \cdot S\), it follows that \(y \notin V(S)\).
	By applying the same argument repeatedly, we find \(\pi^{-m}(x) \notin S\) for \(m \in \Set{1,\ldots,6}\).
	Since \(H\) is \(r\)-symmetric, this implies \(x \notin V(S) \iff \pi(x) \notin V(S)\), \ie{}, \(V(S)\) is \(r\)-symmetric.
	
	\begin{figure}
		\centering
		\begin{minipage}[t]{0.4\textwidth}
			\centering
			\includegraphics[width=\textwidth]{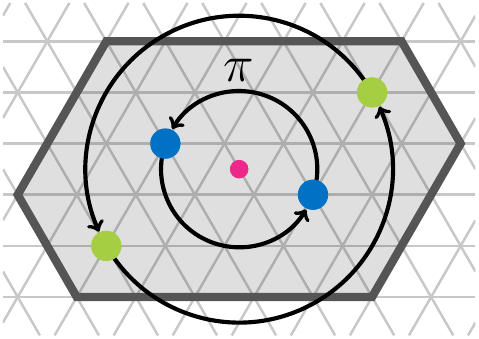}
			\subcaption{The in-place rotation \(\pi\).}
			\label{subfig:rotation_mapping_a}
		\end{minipage}
		\hfill
		\begin{minipage}[t]{0.5\textwidth}
			\centering
			\includegraphics[width=\textwidth]{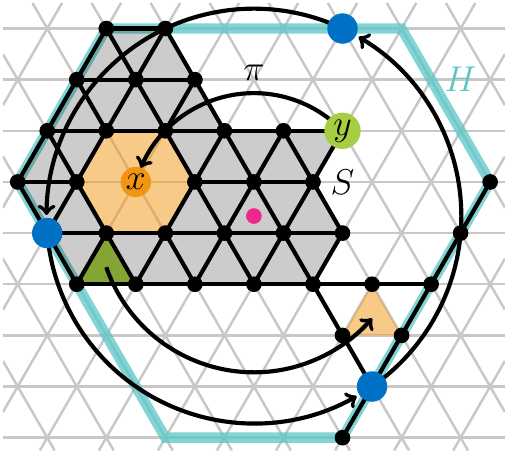}
			\subcaption{\(x \notin V(S)\) implies \(y = \pi^{-1}(x) \notin V(S)\).}
			\label{subfig:rotation_mapping_b}
		\end{minipage}
		\caption{Illustration of the proof of Lemma~\ref{lem:rotation_symmetry}.
			(\subref{subfig:rotation_mapping_a}) shows a \(3\)-symmetric shape with the in-place rotation mapping \(\pi\), demonstrated by the colored nodes and the arrows.
			(\subref{subfig:rotation_mapping_b}) shows a shape \(S\) inside its \(2\)-symmetric convex hull \(H\).
			The shape itself is not symmetric and its highlighted node \(y\) is rotated onto the missing node \(x \in V(H) \setminus S\).
			For sufficiently large \(k\), \(\Map(y)\) will be inside the scaled orange region around \((k+1) \cdot x\), which is not contained in \((k+1) \cdot S\).
			Similarly, the highlighted green face of \(S\) will eventually contain a node that is rotated into the missing orange face.}
		\label{fig:rotation_mapping}
	\end{figure}
	
	Finally, to handle edges and faces of \(S\) that might make it asymmetric, we further restrict the scale to be \(k \geq 3 \cdot d\).
	For such scales, each edge and face of \(S\) is represented by at least one node in \(k \cdot S\), to which we can apply the same arguments as above because the missing edge or face of \(S\) creates a sufficiently large hole in \((k+1) \cdot S\) in which such a node will be placed by the transformation \(\Map\).
\end{proof}

Using these lemmas, we can now prove the main theorem about self-contained and star convex shapes.

\begin{proof}[Proof (Theorem~\ref{theo:star_convex_self_contained})]
	First, let \(S\) be star convex with center node \(c\) and consider two scale factors \(k < k'\).
	By Lemma~\ref{lem:star_convex_helper}, \(S\) can be represented as
	\[
	S = \left(\bigcup\limits_{i=1}^{m_1} P_i \cup \bigcup\limits_{j=1}^{m_2} T_j\right) + c,
	\]
	where the \(P_i\) are parallelograms and the \(T_i\) are Minkowski sums of parallelograms with triangles.
	Then, we have \(k \cdot P_i \subseteq k' \cdot P_i\) and \(k \cdot T_j \subseteq k' \cdot T_j\) since each of the \(P_i\) and \(T_j\) is a convex shape.
	We choose \(t\) such that \(k \cdot c + t = k' \cdot c\), then each of the constituent shapes of \(k \cdot S + t\) is contained in its counterpart in \(k' \cdot S\).
	This already shows that every star convex shape is self-contained.
	
	Now, let \(S\) be a shape that is not star convex.
	We will show that \(S\) is not self-contained by finding a scale \(k \in \Nats\) such that for every \(t \in \Reals^2\) and \(r \in \Set{0, \ldots, 5}\), \((k \cdot S^{(r)} + t) \setminus ((k + 1) \cdot S) \neq \emptyset\).
	Using Lemma~\ref{lem:shortest_path_missing_edge}, we can assume that for every node \(c \in V(S)\), there exists a shortest path \(\Path\) to a node \(v \in V(S)\) such that at least one edge of \(\Path\) is not contained in \(S\).
	If this is not already the case for \(S\), we simply consider \(3 \cdot S\) as the new base shape.
	
	By Lemma~\ref{lem:rotation_symmetry}, we can find \(k_0\) large enough that no non-zero rotation of \(k \cdot S\) fits into \((k+1) \cdot S\) for any scale \(k \geq k_0\) unless \(S\) is rotationally symmetric.
	In this case, however, the rotated version of \(S\) is equivalent to \(S\), so we can disregard rotations altogether because they do not affect the existence of valid translations of \(k \cdot S\) in \((k+1) \cdot S\).
	
	\begin{figure}
		\centering
		\includegraphics[width=0.75\linewidth]{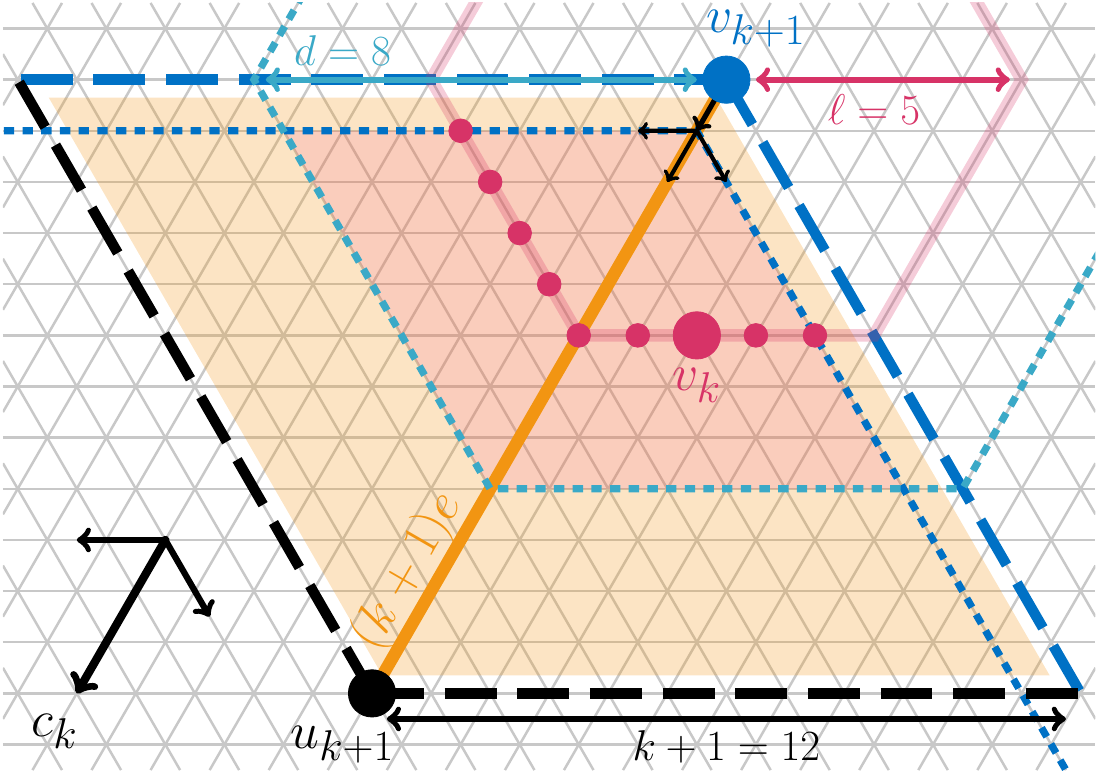}
		\caption{Illustration of the constraints forcing \(v_k\) to lie in a region unoccupied by \((k+1) \cdot S\) for \(k = 11\).
			The scaled unoccupied edge \((k+1) \cdot e\) and the empty parallelogram region it produces are highlighted in orange.
			The parallelogram's edges are drawn as dashed lines of length \(k+1\), incident to \(u_{k+1}\) and \(v_{k+1}\).
			At most these edges of the parallelogram can be occupied by \((k+1) \cdot S\).
			The distance between \(c\) and \(v\) resp.~\(v_k\) and \(v_{k+1}\) is \(\ell = 5 < k + 1\).
			In the bottom left corner, the directions used by all shortest paths from \(v_{k+1}\) to \(v_k\) (resp.~\(v\) to \(c\)) are shown, with the opposite direction of \(e\) emphasized because at least one such edge must be on every path.
			This prevents \(v_k\) from lying on an edge incident to \(v_{k+1}\), shown by the dashed blue lines and the small black arrows.
			Similarly, \(v_k\) cannot lie on any of the dashed black edges incident to \(u_{k+1}\) because it must be closer to \(v_{k+1}\).
			In particular, its distance to \(v_{k+1}\) is bounded by \(d = 8\), as indicated by the light blue dotted lines.
			The red lines show the nodes with distance exactly \(\ell\) to \(v_{k+1}\).
			By combining all constraints, \(v_k\) has to be one of the red nodes in the dark shaded area, none of which are occupied by \((k+1) \cdot S\).}
		\label{fig:star_convex_proof_b}
	\end{figure}
	
	Let \(H\) be the convex hull of \(S\) and let \(d \in \Nats\) be its diameter.
	Consider any scale \(k > d\) and some placement \(c \in \Veqt\) such that \(k \cdot S + c \subseteq (k+1) \cdot S\).
	By Lemma~\ref{lem:fixed_point_properties}, \(c \in V(S)\).
	As argued above, there is a node \(v \in V(S)\) such that a shortest path \(\Path\) from \(c\) to \(v\) has at least one edge that is not contained in \(S\).
	We choose \(v\) and \(\Path\) such that the last edge is missing, \Wlog{}
	Let \(u \in V(H)\) be the predecessor of \(v\) on \(\Path\), \ie{}, the missing edge \(e\) of \(\Path\) connects \(u\) and \(v\).
	Since \(S\) does not contain \(e\), it also does not contain the two incident faces.
	Thus, there is an area in the shape of a regular parallelogram with side length \(k+1\) whose diagonal is \((k+1) \cdot e\) and whose interior does not intersect \((k+1) \cdot S\), \ie{}, only its edges could be covered by edges of \((k+1) \cdot S\) (see Fig.~\ref{fig:star_convex_proof_b}).
	These edges are incident to \(v_{k+1} = (k + 1) \cdot v\) and \(u_{k+1} = (k+1) \cdot u\) and they lie on different axes than \(e\).
	By Lemma~\ref{lem:fixed_point_properties}, the grid distance between \(v_k = k \cdot v + c\) and \(v_{k+1}\) is the length of \(\Path\), which is bounded by \(d\).
	By the choice of \(k > d\), the distance between \(u_{k+1}\) and \(v_{k+1}\) (which is \(k+1\)) is therefore greater than the distance between \(v_k\) and \(v_{k+1}\).
	Now, observe that for any node on one of the parallelogram's edges incident to \(u_{k+1}\), the distance to \(v_{k+1}\) remains \(k+1\).
	Thus, \(v_k\) cannot lie on any of these two edges.
	
	Furthermore, recall from Lemma~\ref{lem:fixed_point_properties} that any shortest path from \(c\) to \(v\) is also a (translated) shortest path from \(v_k\) to \(v_{k+1}\).
	Because such a path contains at least one edge parallel to \(e\), every shortest path from \(v_{k+1}\) to \(v_k\) contains at least one edge in the opposite direction of \(e\).
	Therefore, since the edges of the empty parallelogram that are incident to \(v_{k+1}\) lie on different axes than \(e\), \(v_k\) cannot lie on these two edges either.
	The direction of \(e\) also implies that \(v_k\) must be on the side of the two edges that is closer to \(u_{k+1}\).
	
	Together, these constraints imply that \(v_k\) lies \emph{inside} the parallelogram that is not contained in \((k+1) \cdot S\), so the placement identified by \(c\) does not satisfy \(k \cdot S + c \subseteq (k+1) \cdot S\), contradicting our assumption.
	Since this works for every choice of \(c \in V(S)\), there is no placement of \(k \cdot S\) in \((k+1) \cdot S\).
\end{proof}



\section{The Snowflake Algorithm}
\label{sec:algorithm}

To solve the shape containment problem for snowflake shapes, we present a valid placement search procedure that combines the efficient primitives from Section~\ref{sec:helpers} and embed it into an appropriate scale factor search, employing binary search if possible.


\subsection{Valid Placement Search}
\label{subsec:valid_placement_search}


Let \(A\) be an amoebot structure storing a scale factor \(k\) in some binary counter and let \(S\) be a snowflake represented by a tree \(T = (V_T, E_T)\) with labelings \(\NodeType(\cdot)\), \(\NodeDir(\cdot)\) and \(\NodeLen(\cdot)\).
We assume that each amoebot has a representation of \(T\) and a topological ordering of \(T\) from the leaves to the root in memory.
Such an ordering always exists and can be computed in a preprocessing step.
We compute the valid placements of \(k \cdot S\) in \(A\) as follows:

\begin{enumerate}
    \item For every leaf \(v \in V_T\), perform the placement search for the scaled primitive \(\Line(\NodeDir(v), k \cdot \NodeLen(v))\) or \(\Tri(\NodeDir(v), k \cdot \NodeLen(v))\) represented by \(v\).
    Let \(C(v) \subseteq A\) be the resulting set of valid placements.

    \item Process each non-leaf node \(v \in V_T\) in the topological ordering as follows:
    \begin{enumerate}
        \item If \(\NodeType(v) = \TUnion\), then set \(C(v) = \bigcap_{i=1}^m C(u_i)\), where \(u_1, \ldots, u_m\) are the child nodes of \(v\).
    
        \item If \(\NodeType(v) = \TSum\), let \(u\) be the unique child of \(v\).
        Run the procedure from Lemma~\ref{lem:stretched_placements} to compute \(C(v)\) from \(C(u)\).
    
        \item If \(\NodeType(v) = \TShift\), let \(u\) be the unique child of \(v\).
        Run the procedure from Lemma~\ref{lem:shifted_shapes} to compute \(C(v)\) from \(C(u)\).
    \end{enumerate}
    
    \item Let \(r \in V_T\) be the root of \(T\).
    If \(C(r) \neq \emptyset\), terminate with success and report the valid placements as \(\VP(k \cdot S, A) = C(r)\), otherwise terminate with failure.
\end{enumerate}

To cover the remaining rotations, we simply repeat the procedure five more times and terminate with success if we found at least one valid placement for any rotation.

\begin{lemma}
	\label{lem:snowflake_valid_placements}
	Let \(A\) be an amoebot structure storing a scale \(k\) in some binary counter and let \(S\) be a snowflake.
	Given a description of a tree \(T\) of \(S\) with a topological ordering, the amoebots can compute \(\VP(k \cdot S^{(r)}, A)\) for \(r \in \{0, \ldots, 5\}\) within \(\BigO{\log \min \Set{k, n}}\) rounds.
\end{lemma}
\begin{proof}
	Let \(S\) be represented by the snowflake tree \(T = (V_T, E_T)\).
	Since \(T\) is a tree, a topological ordering always exists.
	All amoebots can store it in memory, encoded as part of their algorithm.
	We prove by induction on \(T\) that \(C(v) = \VP(k \cdot S_v, A)\) for every node \(v \in V_T\), where \(S_v\) is the snowflake represented by \(v\).
	Let \(v\) be a leaf node, then by Lemmas~\ref{lem:line_detection} and \ref{lem:triangle_primitive}, the amoebots compute \(C(v)\) using the line or triangle primitive, identifying all valid placements in \(\BigO{\log \min \Set{k \cdot \NodeLen(v), n}} = \BigO{\log \min \Set{k, n}}\) rounds.
	The side length \(k \cdot \NodeLen(v)\) can be computed on the same counter that stores \(k\) in constant time since \(\NodeLen(v)\) is a constant of the algorithm.
	We assume that the counter is on a segment of maximal length in \(A\) (\eg{}, using all maximal segments of \(A\) as counters).
	If there is not enough space to store \(k \cdot \NodeLen(v)\) on any counter, then \(k\) is too large for any placement and can be rejected immediately.
	
	If \(v\) is a union node, let \(u_1, \ldots, u_m\) be its child nodes and observe that \(C(u_i) = \VP(k \cdot S_{u_i}, A)\) has already been computed for \(1 \leq i \leq m\) due to the topological ordering.
	Since \(S_v\) is the union of all \(S_{u_i}\), \(k \cdot S_v\) is the union of all \(k \cdot S_{u_i}\) and we have
	\begin{align*}
		p \in \VP(k \cdot S_v, A)
		&\iff
		V(k \cdot S_v + p) \subseteq A
		\iff
		\forall\,u_i: V(k \cdot S_{u_i} + p) \subseteq A\\
		&\iff
        \forall\,u_i: p \in \VP(k \cdot S_{u_i}, A)
        \iff
		p \in \bigcap\limits_{i=1}^m \VP(k \cdot S_{u_i}, A).
	\end{align*}
	Thus, \(C(v) = \VP(k \cdot S_v, A)\) is computed correctly and in constant time since each amoebot \(p \in A\) already knows whether \(p \in C(u_i)\) for \(1 \leq i \leq m\) and can perform this step locally.
	
	Next, let \(v\) be a sum or a shift node with unique child node \(u\) and labeled with \(d = \NodeDir(v)\) and \(\ell = \NodeLen(v)\).
	If \(v\) is a sum node, we have \(k \cdot S_v = k \cdot (S_u \oplus \Line(d, \ell))\).
	Since \(C(u) = \VP(k \cdot S_u, A)\) has already been computed, the amoebots can determine \(C(v) = \VP(k \cdot S_v, A)\) within \(\BigO{\log \min \Set{k \cdot \ell, n}}\) rounds, by Lemma~\ref{lem:stretched_placements}.
	If \(v\) is a shift node, we have \(k \cdot S_v = k \cdot ((S_u + \ell \cdot \UVec{d}) \cup \Line(d, \ell))\).
	Using Lemma~\ref{lem:shifted_shapes}, the amoebots can use \(C(u)\) to compute \(C(v) = \VP(k \cdot S_v, A)\) within \(\BigO{\ell \cdot \log \min \Set{k \cdot \ell, n}}\) rounds.
	
	By induction, we have \(C(r) = \VP(k \cdot S, A)\), where \(r\) is the root node of the snowflake tree \(T\) representing \(S\).
	The procedure can be repeated a constant number of times to cover all rotations.
    Since the size of \(T\) and the values of \(\ell\) are constants encoded in the algorithm, the overall runtime is \(\BigO{\log \min \Set{k, n}}\) rounds.
\end{proof}


Using the techniques from the main snowflake algorithm, we now show how to find the valid placements of triangle shapes.

\begin{proof}[Proof (Lemma~\ref{lem:triangle_primitive})]
	Consider an amoebot structure \(A\) that stores a side length \(\ell\) in a binary counter.
	We show that the valid placements of triangles can be constructed in a similar way to those of snowflakes, without using triangle primitives.
	Let \(d_1, d_2 \in \Directions\) be the directions of the two edges incident to the origin of \(\Tri(d_1, \ell)\).
	The node set of the triangle can be written as
	\[
		V(\Tri(d_1, \ell)) = \SetBar{i \cdot \UVec{d_1} + j \cdot \UVec{d_2}}{0 \leq i \leq \ell, 0 \leq j \leq \ell - i}~(\text{with}~i, j \in \Nats_0).
	\]
	We define \(\ell' := \Floor{\ell / 2}\) and \(\ell'' := \ell' + (\ell \bmod 2)\).
	Thereby, we get \(\ell = \ell' + \ell''\).
	Now, we construct three parallelograms whose union covers the node set of the triangle.
	The first parallelogram is \(P_1 := \Line(d_1, \ell') \oplus \Line(d_2, \ell')\).
	We obtain the node set
	\[
		V(P_1) = \SetBar{i \cdot \UVec{d_1} + j \cdot \UVec{d_2}}{0 \leq i \leq \ell', 0 \leq j \leq \ell'}.
	\]
	Then, we have \(V(P_1) \subset V(\Tri(d_1, \ell))\) since \(j \leq \ell - i\) is maintained for all \(0 \leq i \leq \ell'\) due to \(\ell - i \geq \ell - \ell' \geq \ell'\).
	Let \(d_3\) be the direction at \(60^\circ\) counter-clockwise to \(d_2\), then we define the second parallelogram as \(P_2 := (\Line(d_1, \ell') \oplus \Line(d_3, \ell'')) + \ell'' \cdot \UVec{d_1}\).
	It covers the node set
	\[
		V(P_2) = \SetBar{i \cdot \UVec{d_1} + j \cdot \UVec{d_3}}{\ell'' \leq i \leq \ell, 0 \leq j \leq \ell''}.
	\]
	Due to \(\UVec{d_3} = \UVec{d_2} - \UVec{d_1}\), we can rewrite this as
	\begin{align*}
		V(P_2)
		&= \SetBar{(i - j) \cdot \UVec{d_1} + j \cdot \UVec{d_2}}{\ell'' \leq i \leq \ell, 0 \leq j \leq \ell''}\\
		&= \SetBar{i \cdot \UVec{d_1} + j \cdot \UVec{d_2}}{0 \leq j \leq \ell'', \ell'' - j \leq i \leq \ell - j}.
	\end{align*}
	Again, we have \(V(P_2) \subset V(\Tri(d_1, \ell))\) since \(j \leq \ell - i \iff i \leq \ell - j\).
	Finally, let \(d_3'\) be the opposite direction of \(d_3\), then we define the third parallelogram as a mirrored version of the second, \(P_3 := (\Line(d_2, \ell') \oplus \Line(d_3', \ell'')) + \ell'' \cdot \UVec{d_2}\).
	Its node set can be written as
	\begin{align*}
		V(P_3)
		&= \SetBar{i \cdot \UVec{d_2} - j \cdot \UVec{d_3}}{\ell'' \leq i \leq \ell, 0 \leq j \leq \ell''}\\
		&= \SetBar{(i - j) \cdot \UVec{d_2} + j \cdot \UVec{d_1}}{\ell'' \leq i \leq \ell, 0 \leq j \leq \ell''}\\
		&= \SetBar{i \cdot \UVec{d_2} + j \cdot \UVec{d_1}}{0 \leq j \leq \ell'', \ell'' - j \leq i \leq \ell - j}\\
		&= \SetBar{i \cdot \UVec{d_1} + j \cdot \UVec{d_2}}{0 \leq i \leq \ell'', \ell'' - i \leq j \leq \ell - i}.
	\end{align*}
	By symmetry, \(V(P_3) \subset V(\Tri(d_1, \ell))\) holds, so we have \(V(P_1) \cup V(P_2) \cup V(P_3) \subseteq V(\Tri(d_1, \ell))\).
	To show the opposite inclusion, consider some node \(v = i \cdot \UVec{d_1} + j \cdot \UVec{d_2} \in V(\Tri(d_1, \ell))\).
	If \(0 \leq i \leq \ell'\) and \(0 \leq j \leq \ell'\), then \(v \in V(P_1)\).
	If \(i > \ell'\), then \(i \geq \ell''\) and \(j < \ell - \ell' = \ell''\), so \(v \in V(P_2)\).
	Finally, if \(j > \ell'\), then \(j \geq \ell''\) and \(i \leq \ell - \ell'' = \ell' \leq \ell''\), so \(v \in V(P_3)\).
	Together, this implies that \(V(\Tri(d_1, \ell)) = V(P_1) \cup V(P_2) \cup V(P_3)\).
	Although not all edges and faces are covered by this construction, the node set at side length \(\ell\) is the same, so we get the same valid placements and can assume
	\[
		\Tri(d_1, \ell) =
		P_1
		\cup
		\left(P_2 \cup \Line(d_1, \ell'')\right)
		\cup
		\left(P_3 \cup \Line(d_2, \ell'')\right).
	\]
	Observe that the shapes \(P_1\), \(P_2 \cup \Line(d_1, \ell'')\) and \(P_3 \cup \Line(d_2, \ell'')\) are snowflakes that can be constructed without triangle shapes (\eg{}, \(P_2 \cup \Line(d_1, \ell'') = ((\Line(d_1, \ell') \oplus \Line(d_3, \ell'')) + \ell'' \cdot \UVec{d_1}) \cup \Line(d_1, \ell'')\)).
	However, we cannot simply apply the snowflake algorithm because their description size is variable and would cause the runtime to have a linear factor in \(\ell'\) and \(\ell''\).
	To avoid this problem, observe that the occurring lengths \(\ell'\) and \(\ell''\) only differ by at most \(1\), so we can almost treat the shape as if it had only length \(1\) primitives and was scaled by \(\ell'\).
	In the case where \(\ell\) is even, we have \(\ell' = \ell''\) and the snowflake procedure can be applied directly since
	\[
		((\Line(d_1, \ell') \oplus \Line(d_3, \ell'')) + \ell'' \cdot \UVec{d_1}) \cup \Line(d_1, \ell'')
		=
		\ell' \cdot ((\Line(d_1, 1) \oplus \Line(d_3, 1) + \UVec{d_1}) \cup \Line(d_1, 1)).
	\]
	Otherwise, if \(\ell'' = \ell' + 1\), observe that \(\Line(d_1, \ell') \oplus \Line(d_3, \ell'')\) has a minimal axis-width of \(\ell'\) on the axis parallel to \(d_1\).
	Therefore, the shift procedure used by Lemma~\ref{lem:shifted_shapes} can be applied first with distance \(\ell'\) and then again with distance \(1\), shifting the valid placement information a single position further.
	The same approach can be applied to \(P_3\) by symmetry.
	Overall, this takes \(\BigO{\log \min \Set{\ell', n}} = \BigO{\log \min \Set{\ell, n}}\) rounds.
	
	At the end, we intersect the valid placement sets to obtain the valid placements of \(\Tri(d_1, \ell)\).
	Note that even though the shapes we used for this construction depend on \(\ell\), the amoebots never need explicit representations of these shapes (which could exceed the constant memory constraint) to determine their valid placements.
	They simply apply the primitives with the inputs \(\ell'\) and \(\ell''\), which can be computed easily from \(\ell\).
\end{proof}


\subsection{Final Algorithm}
\label{subsec:final_algorithm}

To assemble the final algorithm, we embed the valid placement search procedure into a scale factor search.
Depending on whether the target shape is star convex, we either apply a binary search or use the triangle primitive to determine a small upper bound for a linear search.

\begin{theorem}
    \label{theo:star_convex_shapes_algo}
    Let \(A\) be an amoebot structure and \(S\) a snowflake shape.
    Given a tree representation of \(S\), the amoebots can compute \(k = \kmax(S, A)\) in a binary counter and determine \(\VP(k \cdot S^{(r)}, A)\) for all \(r \in \Set{0, \ldots, 5}\) within \(\BigO{\log^2 k}\) rounds if \(S\) is star convex and \(\BigO{K \log K}\) rounds otherwise, where \(k \leq K = \kmax(\Tri(\E, 1), A) = \BigO{\sqrt{n}}\).
\end{theorem}
\begin{proof}
	The amoebots first establish binary counters on all maximal segments in \(A\), for all axes.
	At least one of these will be large enough to store \(k = \kmax(S, A)\) as long as \(S\) is non-trivial.
	In the following, they use all of these counters simultaneously and deactivate every counter exceeding its memory during an operation.
	
	Consider the case where \(S\) is star convex.
	By Lemma~\ref{lem:binary_search}, combined with Lemma~\ref{lem:snowflake_valid_placements}, the amoebots can compute \(\kmax\) using a binary search and find all valid placements at all six rotations within \(\BigO{\log^2 k}\) rounds.
	Now, let \(S\) be a snowflake that is not star convex.
	In this case, \(S\) must contain at least one triangular face because all snowflakes without faces are unions of lines meeting at the origin, which are star convex (observe that the Minkowski sum of an edge with a line on a different axis always contains some faces).
	Then, \(K = \kmax(\Tri(\E, 1), A)\) is an upper bound for \(\kmax(S, A)\).
	The amoebots can compute \(K\) within \(\BigO{\log^2 K}\) rounds and store it in binary counters, according to Corollary~\ref{cor:largest_triangles}.
	A simple linear search for \(\kmax\) yields the runtime of \(\BigO{K \log K}\) by Lemma~\ref{lem:linear_search}.
	\(K = \BigO{\sqrt{n}}\) follows from the fact that the number of nodes covered by \(k \cdot \Tri(\E, 1)\) grows quadratically with \(k\).
\end{proof}

To make our shape containment algorithms more useful for subsequent procedures, we also give an algorithm that constructs one placement by identifying the amoebots it covers.

\begin{theorem}
    \label{theo:shape_construction}
    Let \(S\) be an arbitrary shape and \(A\) an amoebot structure that stores \(1 \leq k \leq \kmax(S, A)\) in a binary counter and has a description of \(S\).
    Let \(p \in \VP(k \cdot S^{(r)}, A)\) for a rotation \(r \in \Set{0, \ldots, 5}\) be chosen by the amoebots.
    Then, the amoebots can compute \(V(k \cdot S^{(r)} + p)\) within \(\BigO{\log k}\) rounds and at the end of the procedure, each amoebot knows which node, edge or face of the original shape it represents.
\end{theorem}
\begin{proof}
	Let \(S\) be arbitrary and consider an amoebot structure \(A\) that stores a scale \(1 \leq k \leq \kmax(S, A)\) in a binary counter.
	Suppose a valid placement \(p \in \VP(k \cdot S^{(r)}, A)\) has been elected and all amoebots know \(S\) and \(r\).
	
	Now, let \(V\) be the set of grid nodes covered by \(S\), let \(E\) be the set of edges and \(F\) the set of faces.
	The number of elements \(\Abs{V} + \Abs{E} + \Abs{F}\) is constant, so each amoebot can store the information which of the parts it belongs to, if any.
	Consider a sequence of directed edges \(E' = (e_1, \ldots, e_m)\) such that every edge \(e \in E\) is represented by a directed edge in \(E'\) and each \(e_i = (u, v)\) starts at a node \(u\) that is either the origin or has already been reached by \(e_j = (w, u)\) with \(j < i\).
	Such a sequence always exists because \(S\) is connected.
	
	We construct the set \(V(k \cdot S^{(r)} + p)\) by traversing the edges in \(E'\).
	For simplicity, let \(r = 0\); the other rotations are handled analogously.
	First, \(p\) categorizes itself as the origin node in \(V\).
	Consider the next edge \(e = (u, v)\) in \(E'\) and assume that the amoebot \(q = p + k \cdot u\) has already identified itself as \(u\).
	Let \(d\) be the direction from \(u\) to \(v\), then all amoebots \(q + i \cdot \UVec{d}\) for \(1 \leq i \leq k\) must exist because \(p\) is a valid placement of \(k \cdot S^{(r)}\).
	The amoebots establish maximal segments on the axis parallel to \(d\) and \(q\) beeps on a circuit facing direction \(d\) to alert this part of the segment.
	These alerted amoebots now run the \PASC{} algorithm with \(q\) as the start point while \(k\) is transmitted on the global circuit.
	Each amoebot \(q_i = q + i \cdot \UVec{d}\) for \(1 \leq i \leq m\) receives the bits of \(i\) and compares \(i\) to \(k\).
	The amoebots \(q_i\) with \(i < k\) categorize themselves as representatives of the edge \(e\) and the unique amoebot with \(i = k\) categorizes itself as \(v\).
	This takes \(\BigO{\log k}\) rounds (Lemma~\ref{lem:pasc_cutoff}) and covers all amoebots representing the edge \(e\).
	We repeat this for all edges in \(E'\).
	Due to the construction of \(E'\) and since \(p\) is a valid placement, the procedure covers all edges and correctly identifies the representatives of all nodes and edges of \(S\).
	
	Finally, to identify the amoebots in the faces, we construct a circuit for each face.
	Each amoebot representing an edge of \(S\) establishes two partition sets, one for each face incident to the edge (regardless of whether it is contained in \(S\)).
	Each of these partition sets connects the pins facing two adjacent neighbors of the amoebot.
	All other amoebots connect all of their pins in one partition set as if to establish a global circuit.
	Now, in the following \(\Abs{F}\) rounds, the edge amoebots beep on the partition sets facing the incident faces in \(F\), reserving one round for each face.
	Because they know which edge they belong to, the amoebots know in which round to beep.
	The amoebots not assigned to any elements of \(S\) so far identify themselves with the corresponding face if they receive this beep.
	Since \(\Abs{F}\) is constant, this takes only a constant number of rounds.
\end{proof}

Note that after running any shape containment algorithm that identifies a scale and a set \(C \subseteq A\) of valid placements, the amoebots can run a leader election within \(\BigO{\log \Abs{C}}\) rounds, \whp{}~\cite{feldmann2022coordinating}, and then construct the shape for the selected placement using Theorem~\ref{theo:shape_construction}.


\section{Conclusion and Future Work}
\label{sec:conclusion}

In this paper, we introduced the shape containment problem for the amoebot model of programmable matter and presented first sublinear solutions using reconfigurable circuits.
We showed that for some shapes, there is a lower bound of \(\BigOmega{\sqrt{n}}\) rounds due to the arbitrary distribution of valid and invalid placements, even if the maximum scale is already known.
Motivated by fast methods of transferring information using circuits, we constructed the class of \emph{snowflake} shapes that can be solved in sublinear time.
For the subset of shapes that are \emph{star convex}, we even showed how to solve the problem in polylogarithmic time and proved that binary search is not generally applicable to other shapes.

It would be interesting to explore whether the lower bound can be improved when considering the whole problem, where the scale factor is not known.
Naturally, efficient solutions for arbitrary shapes or other shape classes are of interest as well.
To support shape formation algorithms, the related problems of finding the smallest scale at which a given shape contains the amoebot structure and finding scales and placements with maximal overlap and minimal difference could be investigated.
To expand the set of possible applications further, one could examine other, non-uniform scaling behaviors, perhaps allowing the shapes to maintain fine details at larger scales, similar to fractal shapes.


\bibliography{bibliography.bib}

\end{document}